\documentclass[review,3p,authoryear]{elsarticle}
\usepackage{hyperref}
\usepackage{adjustbox}
\usepackage{graphicx} 
\usepackage{longtable}
\usepackage{booktabs}
\usepackage{appendix}
\usepackage{pdflscape}
\usepackage{chngcntr}
\usepackage[labelfont=bf]{caption}
\usepackage[flushleft]{threeparttable}
\usepackage[online]{threeparttablex}
\usepackage{siunitx}
\usepackage{amsthm, amssymb, amsmath}
\usepackage{bbm}
\usepackage{marginnote}
\usepackage{tikz} 
\usepackage{pgf}
\usepackage{tkz-euclide}
\usepackage{xcolor}

\newtheorem{theorem}{Theorem}[section]

\counterwithin{claim}{section}
\newtheorem{observation}{Observation}[section]
\counterwithin{observation}{section}
\newtheorem{corollary}{Corollary}[section]
\counterwithin{corollary}{section}

\counterwithin{lemma}{section}
\newtheorem{remark}{Remark}[section]
\counterwithin{remark}{section}
\newtheorem{example}{Example}[section]
\counterwithin{example}{section}
\usetikzlibrary{automata,fit,positioning}
\usetikzlibrary{calc,positioning,shapes.geometric}
\usetikzlibrary{arrows,shapes,backgrounds}
\usetikzlibrary{petri}
\usepackage{enumerate}
\usepackage{algorithm}
\usepackage[noend]{algpseudocode}
\usepackage{multirow}
\usepackage{xr}

\makeatletter
\@addtoreset{algorithm}{section}
\makeatother

\renewcommand{\thealgorithm}{\thesection.\arabic{algorithm}}

\makeatletter
\def\BState{\State\hskip-\ALG@thistlm}
\makeatother

\makeatletter
\newenvironment{breakablealgorithm}
{
		\begin{center}
			\refstepcounter{algorithm}
			\hrule height.8pt depth0pt \kern2pt
			\renewcommand{\caption}[2][\relax]{
				{\raggedright\textbf{\fname@algorithm~\thealgorithm} ##2\par}%
				\ifx\relax##1\relax 
				\addcontentsline{loa}{algorithm}{\protect\numberline{\thealgorithm}##2}%
				\else 
				\addcontentsline{loa}{algorithm}{\protect\numberline{\thealgorithm}##1}%
				\fi
				\kern2pt\hrule\kern2pt
			}
		}{
		\kern2pt\hrule\relax
	\end{center}
}
\makeatother
\usepackage{amsmath,xparse}
\makeatletter
\NewDocumentCommand{\lplabel}{o m}{%
	\makebox[0pt][r]{(#2)\hspace*{2em}}%
	\IfNoValueF{#1}
	{\def\@currentlabel{#2}\ltx@label{#1}}
}
\makeatother

\begin{document}
\begin{frontmatter}
\title{The Project Scheduling Interdiction Problem with Delay Groups}

\author[inst1]{Fei Wu}

\affiliation[inst1]{organization={Research Center for Operations Management},
            addressline={KU Leuven}, 
            country={Belgium}}

\author[inst1]{Erik Demeulemeester}
\author[inst1]{Jannik Matuschke}
\begin{abstract}

Large-scale projects are frequently delayed by correlated disruptions: when a shared input such as a common supplier, a specialized team, or a supporting platform degrades, all dependent activities are slowed down simultaneously. This paper introduces
delay groups to capture such disruptions: a delay group is a set of activities whose delays stem from a common cause, described jointly
by an uncertainty set. Our model takes the perspective of an interdictor that, subject to a budget of $k$ groups, selects which groups to disrupt so as to maximize the project makespan. The interdictor can extend activity durations within each disrupted
group by delays from a group-specific uncertainty set, while non-disrupted activities keep their nominal duration. We study the complexity of the resulting Project Scheduling Interdiction Problem
with Delay Groups (\textsc{PSIP-DG}), which provides a worst-case stress test of the schedule. The problem is computationally intractable ($N\!P$-hard) for general polyhedral uncertainty sets, even for a single delay group with a continuous knapsack constraint.
For budgeted uncertainty sets, we prove $N\!P$-hardness both when all activities of a disrupted group are delayed and when only one activity per group may be delayed, and derive an inapproximability
bound of $1-1/e+\epsilon$ for the former case. We further develop a greedy heuristic with approximation guarantee $k$ and two structure-based heuristics with initializations and neighborhoods from tractable special cases. Experiments on $5{,}000$-activity
networks show that the heuristics match the solution quality of an exact solver at substantially lower running times, in some cases finding strictly better solutions.
\end{abstract}

\end{frontmatter}

\section{Introduction}
\label{sec:intro}

Project scheduling is a part of project management that determines the sequence or scheduling plan for a series of related activities that are constituents of the project. In the absence of resource constraints, the critical path determines the makespan of the project \citep{goldratt1997critic}. Delays of individual activities have
therefore long been the natural unit of analysis when assessing a schedule's robustness. In large-scale projects, however, a disruption typically affects several activities at once. Activities often share a common enabler, a supplier delivering critical components, a team of specialists supervising several tasks, or a platform supporting parallel operations, and when that enabler degrades, all dependent activities are slowed down together. The 2020–2023 global semiconductor shortage illustrates this at an industry scale: as the supply of individual chip families tightened, manufacturers did not halt production but operated at a reduced pace, with component lead times prolonged by several months \citep{mohammad2022global}. The development of the Boeing 787 shows the same mechanism within a single project. Its assembly relied on approximately fifty tier-1 strategic suppliers, and shortages of shared components, most notably fasteners, simultaneously delayed the assembly activities depending on them; these correlated slowdowns accumulated along the precedence network, ultimately amounting to years of delay and billions in cost
overruns \citep{tang2009managing}.

Two features of these episodes matter for modeling. First, the affected activities were slowed down rather than stopped, so the disruption manifests as an increase in processing times. Second, the delays were correlated: activities sharing a dependency were hit together, though to different extents. Robust project scheduling models typically treat activity delays as independent perturbations. Under a common-cause disruption, however, several activities on the critical path, or on parallel paths, can be delayed at once, so an independence assumption can substantially understate the worst-case completion time. To capture this structure, we introduce delay groups: sets of activities whose delays are driven by a common underlying cause, with the joint delay within each group described by an uncertainty set.

We study this correlation structure in a precedence-constrained setting without resource constraints, so that any computational difficulty arises from the delay groups alone: our hardness results show that the problem is already $N\!P$-hard and hard to approximate in this basic setting, so these negative results extend to more complex settings, including resource-constrained variants. Moreover, the disruptions we model, namely the degradation of a shared enabler, slow down all dependent activities simultaneously rather than forcing them to queue for capacity, and are therefore naturally captured by processing time increases rather than by resource constraints. 

Neither of these disruptions was caused by a deliberate attacker; the semiconductor shortage and Boeing's supplier failures arose from pandemic shocks and coordination failures. Their lesson, however, is best understood adversarially: the damage came from the combination of correlated delays, and the worst-case combination was never assessed in advance. Evaluating each shared dependency in isolation is insufficient, as the joint impact of simultaneous disruptions on the critical path can far exceed the sum of their individual effects. Identifying it is not a matter of simple enumeration: the difficulty stems both from selecting which groups to disrupt and from allocating delays within each group, and our complexity results show that the problem is $N\!P$-hard even in restricted settings, such as a single delay group or groups of at most two activities. 

We therefore take an interdiction perspective: a hypothetical adversary, allowed to disrupt at most $k$ delay groups with delays chosen from the uncertainty sets $Q_r$, selects the combination to maximize the project makespan. The resulting worst-case makespan serves as a stress test of the schedule: it quantifies how much the completion time can deteriorate under a bounded number of correlated disruptions. \citet{szmerekovsky2023project} argue that such worst-case risk assessment enables organizations to prepare not only for failures catastrophic to the project, but also for those catastrophic to the firm. The budget $k$ and the uncertainty sets jointly keep the analysis from becoming overly conservative: bounding the number of simultaneously disrupted groups follows the budget-of-uncertainty principle that is standard in robust scheduling \citep{bertsimas2004price, bruni2017adjustable, bold2021compact}. Modeling delays via uncertainty sets rather than via probability distributions is also motivated by informational considerations. For one-off projects such as new aircraft development, the joint probability distribution of such delays is practically impossible to estimate. What project managers can assess, however, is which activities share a common dependency and the plausible range of the resulting delays, and these are exactly the inputs our model requires. The output of the model is directly actionable: it identifies the $k$ most vulnerable delay groups and the associated worst-case completion time, indicating where robustness investments such as backup suppliers, redundant expertise, or schedule buffers yield the greatest protection. We call the resulting problem the Project Scheduling Interdiction Problem with Delay Groups (\textsc{PSIP-DG}).

\paragraph{Literature review} Robust optimization has been extensively studied for scheduling problems under polyhedral or budgeted uncertainty, building on the budget-of-uncertainty framework of \citet{bertsimas2004price} introduced above. Subsequent work by \citet{bertsimas2011theory} and \citet{bertsimas2022robust} extends this framework by formalizing a unified robust optimization approach and deriving tractable formulations for common uncertainty sets, including box, ellipsoidal, budgeted, and polyhedral sets. Building on this foundation, many studies have explored robust scheduling under budgeted uncertainty, with much of the research focusing on machine scheduling environments \citep{juvin2025flow}; see \citet{herroelen2005project} for a survey of project scheduling under uncertainty. In the context of project scheduling, \citet{bruni2017adjustable} introduce a two-stage robust
resource-constrained project scheduling problem with uncertain activity durations, solved via Benders decomposition. \citet{bold2021compact} later reformulate the inner adversarial problem as a longest-path problem and derive a compact MILP through duality, outperforming Benders-based methods in computational experiments; this formulation is subsequently extended to a multi-mode setting \citep{bold2022faster}. While these robust resource-constrained models capture how scarce resources couple activities \citep{lambrechts2008proactive}, correlated delays arising from a common cause have not been modeled.

Beyond uncertainty, adversarial and interdiction-based perspectives have received growing attention. \citet{brown2005complexity} study the computational complexity of several project interdiction models and show that the special case of \textsc{PSIP-DG} with singleton groups (\textsc{PSIP}) can be solved via dynamic programming, whereas the problem becomes strongly $N\!P$-hard if the decision maker is allowed to expedite certain
tasks. \citet{brown2009interdicting} formulate the project interdiction problem as a Stackelberg game between a project manager and an attacker who aims to delay project completion; their bilevel optimization approach illustrates how limited interdiction resources can be allocated to maximally delay the project by targeting critical activities, which amounts to a longest-path interdiction problem on the project network. In these models, interdiction acts on individual activities, with the budget limiting the number or cost of activities that can be delayed; in \textsc{PSIP-DG}, by contrast, the interdictor attacks entire delay groups, whose delays are jointly constrained by an uncertainty set, so that a single interdiction decision affects several activities in a correlated manner.

In line with this work, \citet{pinker2013managing} propose a secret project management model in which the project manager strategically invests resources in deception decisions to minimize the exposed time. They show that this problem is strongly $N\!P$-hard in the general sense and polynomially solvable when no deception is allowed. \citet{pinker2014complexity} further explore the computational
complexity of this model, establishing hardness results and identifying tractable special cases. \citet{gutin2015interdiction} investigate an interdiction game on a PERT network with stochastic (Markovian) task durations, formulating both static and dynamic strategies for the interdictor; under a memoryless assumption on
activity durations, they derive tractable solution methods for each case. Unlike this stochastic model, which requires distributional assumptions on the task durations, our formulation, which is based on uncertainty sets, only requires support information on the delays. \citet{klastorin2013optimal} develop a stochastic dynamic programming approach to optimize slack allocation and activity compression in anticipation of a single disruptive event.  However, none of these works consider disruptions that are coupled across groups of tasks. To the best of our knowledge, this paper is the first to introduce delay groups into project scheduling interdiction models.

Our work is closely related to network interdiction problems. This class of problems is well understood for the case in which the delay scenarios are described by a bound on the total number of disruptions or on the cost of increasing individual durations \citep{brown2009interdicting, israeli2002shortest, sefair2016dynamic, wood2010bilevel}; for reviews of the deterministic and stochastic network interdiction literature, we refer to \citet{wood2010bilevel}, \citet{smith2020survey}, and \citet{morton2010stochastic}. The problem we study is a longest-path interdiction problem, whereas most of the existing literature concerns shortest-path interdiction.
\citet{fulkerson1977maximizing} consider a shortest-path interdiction model in which the length of each arc can be increased continuously and show that it can be solved by parametric linear programming, equivalent to a minimum-cost network flow problem, in polynomial time. When interdiction decisions are binary, however, the problem becomes
$N\!P$-hard \citep{ball1989finding}, and \citet{boros2006inapproximability} provide inapproximability bounds for several variants.

Another relevant problem is the $N\!P$-hard multiple-constraint shortest path (MCSP) problem, which aims to find the least-cost path between a pair of nodes such that the sum of additive edge weights does not exceed a specified limit \citep{garey1979computers}. Its single-constraint special case has been studied extensively, with
numerous exact approaches, approximation schemes, and heuristics available in the literature \citep{lozano2013exact, pugliese2013survey, chen2008two, kuipers2002overview}. In \textsc{MCSP}, however, the constraints restrict the chosen path directly; in \textsc{PSIP-DG}, by contrast, the group-selection variables interact with the path through the bilinear objective.

\emph{Our contribution and outline}. In Section~\ref{sec:definition}, we propose a new model for the project scheduling problem with delay groups and provide formal definitions for the uncertainty set within each group, including polyhedral and budgeted uncertainty sets. In Section~\ref{sec:complex}, we identify some tractable cases of \textsc{PSIP-DG}, including the special case of \textsc{PSIP-DG}, in which each group has only one activity(PSIP). Prior work has shown that PSIP becomes $N\!P$-hard when interdiction costs are non-uniform and interdiction decisions are discrete, but remains polynomially solvable when interdiction levels are continuous under uniform costs. Our work fills a gap in the literature by analyzing a variant that simultaneously considers continuous interdiction decisions and non-uniform interdiction costs: we prove that this continuous PSIP with non-uniform interdiction costs is $N\!P$-hard, and this result directly implies the $N\!P$-hardness of \textsc{PSIP-DG}. Using two distinct reductions from the \textsc{Max $ k $-cover} and 3-\textsc{SAT} problems, we show that \textsc{PSIP-DG} is $N\!P$-hard even for more restricted uncertain activity durations in the following two special cases: (i) all activities in the group $ V_r $ can be delayed, and (ii) exactly one activity is allowed to be delayed in each group. We further investigate the approximability of the general \textsc{PSIP-DG} problem and show that the problem does not admit a $(1-1/\mathrm{e}+\epsilon)$-approximation for any $ \epsilon>0 $ unless $ P = N\!P$. The inapproximability result also holds for case (i).

For \textsc{PSIP-DG} with budgeted uncertainty sets, we develop approximation and heuristic algorithms in Section~\ref{sec:heuristics}. We present a greedy heuristic with an approximation ratio exactly equal to the interdiction budget $k$, and we show that the joint impact of a set of delay groups is not
determined by their individual contributions. We further propose two structure-based heuristics that
search the two decision spaces of the problem: the subpath reoptimization search (SRS) iteratively improves a candidate longest path, initialized by relaxing the group restrictions, while the group selection search (GSS) iteratively improves the set of attacked groups, initialized by the exact per-group decomposition. In Section~\ref{sec:experiment}, we evaluate both heuristics against an exact MIP approach on networks ranging from $50$ to over $5{,}000$ nodes and up to $200{,}000$ edges. The heuristics provide near-optimal solutions at a fraction of the running time; for large-scale instances, GSS matches the quality of
the best solutions found by Gurobi within the one-hour time limit, in some cases finding strictly better ones. 
Finally, Section~\ref{sec:conclusions} concludes the paper and discusses future research directions.

\section{Problem definition}
\label{sec:definition}      
A project is represented by a \emph{Project Evaluation and Review
Technique} (PERT) network: a directed acyclic graph $G=(V,E)$ with $V=\{0,1,\dots,n,n+1\}$, where each node is an activity and each arc $(i,j)\in E$ represents a precedence constraint, meaning that activity $j$ can only start once activity $i$ is completed. Nodes $0$ and $n+1$ are
dummy activities marking the start and the end of the project. Given a vector of activity durations $d\in\mathbb{R}^{V}_{+}$, the length of a path $P$ is the sum of its node weights $\sum_{v\in V(P)} d_v$, where $V(P)$ denotes the nodes on $P$, and the makespan of the project is the length of a longest path from $0$ to $n+1$; this path is the \emph{critical path}. For any integer $m$, we use $[m]$ to denote the set $\{1,\dots,m\}$.
 
  We now formally introduce the model for \textsc{PSIP-DG}. We assume that the set of activities $V$ is partitioned into $m$ subsets $V_1,\dots,V_m$, which we call \emph{delay groups}. We denote the family of delay groups by $\mathcal{V} := \{V_r \subseteq V : r \in [m]\}$. Without loss of generality, we assume that the delay groups are disjoint (i.e., they do not intersect). If there is an activity that belongs to different delay groups, we can create duplicate copies of the activity scheduled in a sequential order and assign each copy to a different delay group. All predecessor activities are connected to the first copy and all successor activities are connected to the last copy. The duration of the original activity is assigned to one of the copies, while the durations of all other copies are set to zero. This transformation does not affect the project makespan.

   For every delay group $V_r$, $r \in [m]$, there is an associated uncertainty set $Q_r \subseteq \mathbb{R}^{V_r}_{+}$ that describes all feasible delay vectors for the activities in that group. For each $r \in [m]$, let $x^r \in Q_r$ denote a delay vector on $V_r$, and let $z_r \in \{0,1\}$ indicate whether delay group $V_r$ is attacked by the interdictor. The prolonged activity durations are then given by
\[
  \tilde d_i \;=\; d_i + \sum_{r \in [m]:\, i \in V_r} x^r_i z_r, \qquad i \in V,
\]
where $d_i$ is the nominal duration of activity $i$.
The task of the interdictor is to find a subset $ S \subseteq [m]$ with $ |S| \leq k $, where $k$ denotes the interdiction budget, i.e., the maximum number of activity groups that the interdictor is allowed to attack, and delay vectors $x^r \in Q_r$ for
$r \in S$, so that the length of a longest path from $0$ to $n+1$ with respect to the activity durations $\tilde d$ is maximized.
The main parameters and decision variables are given as follows:

\emph{Parameters:}
$G = (V,E)$, $d_i$, $\mathcal{V} = \{V_r\}_{r \in [m]}$, $Q_r$, budget $k$.

\emph{Decision variables:}
$z_r \in \{0,1\}$ for $r \in [m]$, $x^r \in Q_r$ for $r \in [m]$, and start times $s_i$ for $i \in V$.

PSIP-DG can be formulated as the following optimization problem: 
\begin{alignat*}{2}
	&\max\limits_{x,z}\min\limits_{s}&&s_{n+1} \notag\\
	\lplabel[PSIPDG]{PSIP-DG}&\mbox{s.t.}& \quad &\begin{aligned}[t]
		s_{j}-s_{i}&\ge d_i + \sum_{r \in [m]:\, i \in V_r} x^r_i z_r, &{}& \forall (i,j)\in E,  \\
		\sum_{r\in [m]}z_r&\leq k, &{}&  \\
		x^r&\in Q_r, &{}&  \forall~r\in [m]\\ 
		z_r&\in \{0,1\}, &{}&  \forall~r\in [m]\\
        s_0 &= 0,\ s_i\ge 0, &{}& \forall~i\in V
	\end{aligned}
\end{alignat*} 
Subsequently, we introduce the formal definition of the uncertainty set.
\paragraph{Polyhedral uncertainty set}
For every delay group $V_r$, $r \in [m]$, we define a polyhedral
uncertainty set
\[
  Q_r \;=\; \{ x^r \in \mathbb{R}^{V_r}_{+} : A_r x^r \le b_r \},
\]
where $A_r$ and $b_r$ are given data for group $V_r$. We write
\textsc{PSIP-DG}$_{\mathrm{P}}$ for \textsc{PSIP-DG} with polyhedral
uncertainty. The following two examples illustrate how $A_r$ and
$b_r$ encode different disruption structures.

\begin{example}\label{ex:box}
If $A_r$ is the identity matrix, then $Q_r = \{x^r : 0 \le x^r_i \le
b_{r,i}\}$ is a box: each activity has an individual delay limit, and
the delays within the group are unrelated.
\end{example}

\begin{example}\label{ex:budget}
Suppose the activities in $V_r$ require supervision by a shared team
of specialists. If the team is temporarily understaffed, the lost
capacity bounds the delay the group can suffer in total, while each
activity $i$ can absorb at most $u_i$ before its remaining work
stalls. This corresponds to
\[
Q_r=\Bigl\{x^r\in\mathbb{R}^{V_r}_{+}:\ \textstyle\sum_{i\in V_r}
x^r_i\le \beta_r,\ x^r_i\le u_i\ \text{for all } i\in V_r\Bigr\},
\]
in which $A_r$ stacks a row of ones on top of the identity matrix. The aggregate constraint couples the delays within the group: the total delay cannot exceed $\beta_r$, so a larger delay of one activity leaves less room for the others. The interdictor distributes this total over the group so as to maximize the impact on the project
makespan.
\end{example}
\paragraph{Budgeted uncertainty set}
Let $u_{i}^r \in \{0,1\}$ indicate whether activity $i$ in group $r$ is delayed by the interdictor. The budgeted uncertainty set bounds the number of delayed activities in each group $ V_r $ by a number $\ell_r$, and the delay time of any activity $ i\in V $ is fixed, denoted by $ \Delta_i$.  For each delay group $V_r$ with $r \in [m]$, the corresponding budgeted
uncertainty set is defined as
\[
  \mathcal{U}_r^B \;=\; \Bigl\{ x^r \in \mathbb{R}^{V_r}_{+} :
  \exists\, u^r \in \{0,1\}^{V_r} \text{ such that }
  \sum_{i\in V_r} u_{i}^r \le \ell_r, \,
  x^r_i = \Delta_i  u_{i}^r \text{ for all } i\in V_r \Bigr\}.
\]

We denote PSIP-DG with budgeted uncertainty by \textsc{PSIP-DG}$_{\mathrm{B}}$. In practice, correlated disruptions frequently occur in two canonical forms: either an entire subsystem is shut down (for example, when a shared utility such as compressed air or a common IT platform becomes unavailable), or at most a single critical activity within
a group is affected (for example, when a unique expert or specialized machine is diverted so that only one of the tasks depending on it is delayed). To capture these two extremes, we consider the following two special cases of this uncertainty set:

\begin{itemize}
	\item \textsc{PSIP-DG}$_{\mathrm{B}, \infty}$: $ \ell_r\geq |V_r| $, for any $r\in[m]$, all the activities in the group $ V_r $ can be delayed.
	\item \textsc{PSIP-DG}$_{\mathrm{B}, 1}$: $ \ell_r= 1 $, for any $r\in[m]$, only one  activity is allowed to be delayed in each group. 
\end{itemize}
 Next, we give the dual formulation of \textsc{PSIP-DG} with budgeted uncertainty,
\begin{alignat*}{2}
	&\max\limits_{u,y,z}&&\sum_{e=(i,j)\in E} y_e\Bigl( d_i + \sum_{r\in [m]:\, i\in V_r} \Delta_i u_i^r \Bigr) \\
\lplabel[dpu]{$\mathrm{DP}(\mathcal{U}^B)$}	&\mbox{s.t.}& \quad 
	&\sum_{i\in V: e=(i,j)}\begin{aligned}[t]y_e-\sum_{i\in V: e=(j, i)}y_e&=\begin{cases} 
		-1& \text{if } j=0 \\
		1 & \text{if } j=n+1\\
		0 & \text{if } j\in V\setminus\{0,n+1\}
	\end{cases}&{}&   \\
\sum_{r\in [m]}z_r&\leq k &{}&   \\
\sum_{i\in V_r} u_{i}^r&\leq \ell_r\cdot z_r  &{}&   \forall~ r\in [m] \\
 y_e, z_r, u_{i}^r&\in\{0,1\}  &{}&  \forall~ e
 \in E,~r\in [m],~ i\in V_r
\end{aligned}
\end{alignat*}
\paragraph{Illustrative example} Consider a small instance of \textsc{PSIP-DG}$_{\mathrm{B}}$ whose project
network is shown in \autoref{fig:illustrative_case}. The project consists of eight non-dummy activities between dummy activity $0$ and dummy activity $9$, where the nominal durations $d_i$ are indicated above
each activity node $i$. Activities $1$, $2$, and $3$ require regular supervision by the same team of experts and therefore form delay group $V_1=\{1,2,3\}$. Activities $4$, $5$, and $7$ process components sourced from a common external supplier and constitute delay group $V_2=\{4,5,7\}$. As shown in the right-hand table of \autoref{fig:illustrative_case}, we apply a budgeted uncertainty set to each delay group: if the expert team for $V_1$ becomes understaffed or the supplier for $V_2$ falls behind on deliveries, all activities in the corresponding group may be slowed down, so we set $\ell_1 = \ell_2 = 3$. Finally, activities $6$ and $8$ install two different components sourced from the same supplier.  A shortage at this supplier typically affects one component line at a time, so at most one activity in $V_3 = \{6,8\}$ is delayed ($\ell_3 = 1$). The delay increments $\Delta_i$ are shown as red numbers above the  nodes. The interdiction budget is set to $k=2$, meaning that at most two delay groups can be attacked simultaneously.
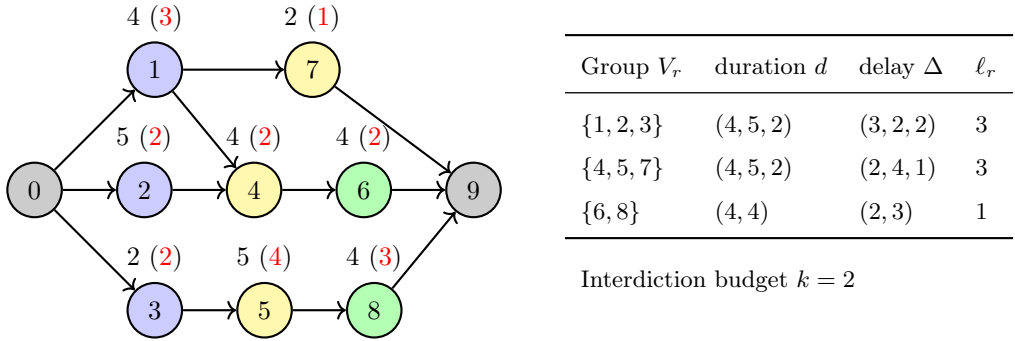
\begin{figure}[htbp]
\centering
\begin{minipage}{0.5\textwidth}
  \centering
  \begin{tikzpicture}[
    thick, node distance=15mm,
    cir/.style={draw, circle, text width=3.2mm, align=center}
  ]
    \node[cir, fill=gray!40] (0) {$0$};
    \node[cir, right=2em of 0, fill=blue!20,  label=above:{$5$ (\textcolor{red}{2})}] (2) {$2$};
    \node[cir, above right=of 0, fill=blue!20, label=above:{$4$ (\textcolor{red}{3})}] (1) {$1$};
    \node[cir, below right=of 0, fill=blue!20, label=above:{$2$ (\textcolor{red}{2})}] (3) {$3$};

    \node[cir, right=3.8em of 1, fill=yellow!40, label=above:{$2$ (\textcolor{red}{1})}] (7) {$7$};
    \node[cir, right=2em of 2, fill=yellow!40, label=above:{$4$ (\textcolor{red}{2})}] (4) {$4$};
    \node[cir, right=2em of 4, fill=green!30,  label=above:{$4$ (\textcolor{red}{2})}] (6) {$6$};
    \node[cir, right=2em of 6, fill=gray!40] (9) {$9$};

    \node[cir, right=2em of 3, fill=yellow!40, label=above:{$5$ (\textcolor{red}{4})}] (5) {$5$};
    \node[cir, right=2em of 5, fill=green!30,  label=above:{$4$ (\textcolor{red}{3})}] (8) {$8$};

    \draw[->] (0) edge (1) (0) edge (2) (0) edge (3);
    \draw[->] (1) edge (4) (2) edge (4)  (3) edge (5) (1) edge (7)
              (4) edge (6) (5) edge (8) (7) edge (9) (6) edge (9) (8) edge (9);
  \end{tikzpicture}
\end{minipage}%
\hfill
\begin{minipage}{0.5\textwidth}
  \small
  \begin{tabular}{llll}
    \toprule
    Group $V_r$& duration $d$ & delay $\Delta$  &$\ell_r$ \\
    \midrule
   $\{1, 2, 3\}$&$(4,5,2)$&$(3,2,2)$ & $ 3$ \\
   $\{4,5,7\}$&$(4,5,2)$ &$(2,4,1)$ &  $3$ \\
   $\{6,8\}$& $(4,4)$ & $(2,3)$ &  $1$ \\
    \bottomrule
  \end{tabular}

  \vspace{0.6em}

  \begin{tabular}{ll}
  Interdiction budget $k=2$ \\
  \end{tabular}
\end{minipage}
\caption{Project network (left) and delay groups
$V_r$, duration vector $d$, delay vector $\Delta$, and interdiction budget $k$ (right).}
\label{fig:illustrative_case}
\end{figure}

In this setting, \textsc{PSIP-DG}$_{\mathrm{B}}$ evaluates which delay groups and activities, if disrupted, would cause the largest increase in the project makespan. Under nominal durations the project makespan is $13$, while an optimal interdiction policy selects groups $V_2$ and $V_3$, delaying all activities in $V_2$ and activity~8 and leading to a worst-case project makespan of $18$ time units, determined by the path $0$--$3$--$5$--$8$--$9$. From a managerial perspective, such a solution simultaneously evaluates the worst-case completion time and identifies where robustness efforts should be invested: for instance, by hiring additional expert capacity or designing backup procedures for activities in $V_2$, and prioritizing activity~8 in $V_3$ for extra buffer, backup resources, or closer monitoring. In this way, the \textsc{PSIP-DG}$_{\mathrm{B}}$ solution can be directly translated into targeted robustness actions on the most vulnerable parts of the project. Interestingly, attacking any single group in isolation yields the same makespan of $15$, yet different combinations of groups lead to different worst-case makespans. The impact of an interdiction policy thus depends on how the selected groups interact through the network, which underlies the computational difficulty established in Section~3.

\section{Complexity results for \textsc{PSIP-DG}}
\label{sec:complex}
This section focuses on the computational complexity of the \textsc{PSIP-DG} under different types of uncertainty sets and network structures. We first present several tractable cases where the \textsc{PSIP-DG}  can be solved efficiently, given a specific structure of  delay groups. Then we show that calculating the optimal policy for each separate delay group is already $N\!P$-hard under polyhedral uncertainty. Furthermore, we prove that the \textsc{PSIP-DG} is  $N\!P$-hard in the two special cases within the budgeted uncertainty set and present an inapproximability result.

\subsection{Tractable cases}
\subsubsection{Singleton delay group}
\label{sec:psip}
We begin with the basic case that each activity constitutes its own group, ignoring the``group" concept. For any $ j\in V $, adding a delay $\Delta_j$ to the duration of activity $ j $ consumes one unit of interdiction budget. The total interdiction budget is $ k $, which means the interdictor can only increase the processing time of $ k $ activities with the budget constraint $\sum_{i\in V} z_i \leq k$. This so-called project scheduling interdiction problem (PSIP) can be solved through dynamic-programming recursion in $\mathcal{O}(k|E|)$ time, where $|E|$ is the number of edges in the project network \citep{brown2005complexity}. There are some variants of PSIP:
\begin{itemize}
	\item Non-uniform interdiction cost: Suppose the interdiction cost for activity $i$ is $k_i$, and the constraint for $x$ is rewritten as $\sum_{i\in V} k_iz_i \leq k$. The $N\!P$-hardness of such a problem can be shown by reduction from the knapsack problem \citep{brown2005complexity}. 
	
	\item Continuous interdiction: The interdiction decisions $z$ are allowed to be continuous variables and the budget $k$ is given as a fractional number.	Based on the following two observations, the linearization technique used in \citet{lim2007algorithms} can be applied to solve continuous PSIP in polynomial time.
	\begin{observation}
		There exists an optimal solution $ (x^*, s^*, z^*) $ to PSIP such that $  \sum\limits_{i\in V} z^*_{i}=k$ \citep{sherali1980finitely}.	
	\end{observation}
	\begin{observation}
		Consider the polyhedral set $ Z=\{z: \sum\limits_{i\in V} z_{i}=k, 0\leq z_i\leq 1\}.$ Then, for each extreme point $z$ of $ Z$, there exists a single basic variable $ Z_r $ such that $z_r\in[0, 1] $, while all other variables are nonbasic at their lower bounds of 0 or upper bounds of 1 \citep{bertsimas1997introduction}. 
	\end{observation}
    \begin{observation}\label{obs:extreme-point}
	There exists an optimal solution $ (x^*, z^*) $ to \ref{PSIPDG} such that $\sum\limits_{i\in V} z^*_{i}=k$. Given the polyhedral set $$Z=\{z: \sum\limits_{i\in V} k_iz_{i} \leq k, 0\leq z_i\leq 1\},$$ for each extreme point $z$ of $ Z$, there exists a single basic variable $ z_r $ such that $z_r\in[0, 1] $, while all other variables are nonbasic and are either 0 or 1.
\end{observation}
	
\end{itemize}
\subsubsection{Delay groups with specific structure}
\label{sec:tractable_cases}
Case 1: One tractable case is the one where each delay group forms an anti-chain. $F$ is termed an \emph{anti-chain} if there are no transitive relations between $u$ and $v$ for all $u, v \in F$. Under the assumption that no delay groups are overlapping, we can discretize each delay group into a separate activity. For each activity $i\in V_j$, calculate the maximum delay { $\Delta_i=\max\{x_i:~ x_i\in Q_j\}$.} This process simplifies PSIP-DG to the setting of a singleton delay group. If each delay group forms an anti-chain, it is equivalent to solving PSIP with interdiction budget $ k $ and the delay for activities is represented as above. 

Case 2: Another interesting special case is the one where there are no precedence constraints among activities in different delay groups in \textsc{PSIP-DG}$_{\mathrm{B}}$. Since we consider an activity-on-node network, it means that there are no paths from any activity in one group to the activities in another. This condition implies that activities in different delay groups are located in different chains in the parallel decomposition of the maximal chains in $G$. As a result, the disruptions to different delay groups are independent of each other. 

We denote the optimal profit of consuming one unit of budget to the group $ V_r $ as $ W_r $. Given the budgeted uncertainty in $V_r$, the optimal value can be obtained by solving the following optimization problem:

\begin{alignat*}{2}
	&\max\limits_{y,z}&&\sum_{e=(i,j)\in E}y_e\left(d_i+\Delta_iz_{i} \right)\\
	\lplabel[dplr]{$\mathrm{MAX}(\ell, r)$}	&\mbox{s.t.}& \quad &\sum_{i\in V: e=(i,j)}\begin{aligned}[t]y_e-\sum_{i\in V: e=(j, i)}y_e&=\begin{cases} 
			-1& \text{if } j=0 \\
			1 & \text{if } j=n+1\\
			0 & \text{if } j\in V\setminus\{0,n+1\}
		\end{cases}&{}&   \\
		\sum_{i\in V_r} z_{i}&\leq \ell_r &{}&   \\
		z_{i}&\in\{0,1\}  &{}&  \forall~ i\in V_r\\
        z_i &= 0 &{}&\forall\, i\notin V_r.
	\end{aligned}
\end{alignat*}

A feasible solution $ \{y_e\}_{e\in E} $ corresponds to an $0$-$(n+1)$-path in the PERT network.  Thus, \ref{dplr} maximizes the longest source-sink path length in an acyclic digraph within a budget, solvable in $\mathcal{O}(\ell_r|E|)$ time. Let $W$ denote the optimal objective value of \textsc{PSIP-DG}$_{\mathrm{B}}$ under such conditions. Since the disruptions to different delay groups are independent, we have $W= \max_{r\in [m]} W_r$, where $W_r$ is obtained by solving the problem \ref{dplr} for the corresponding delay group $r\in [m]$. Therefore, \textsc{PSIP-DG}$_{\mathrm{B}}$ can be solved by solving \ref{dplr}  for each $r\in [m]$ independently and taking the maximum of the optimal values. 

\subsection{\texorpdfstring{$N\!P$-hardness of \textsc{PSIP-DG}$_{\mathrm{P}}$}{NP-hardness of PSIP-DG-P}}
To study the complexity of \textsc{PSIP-DG}$_{\mathrm{P}}$, we first consider the continuous interdiction case with non-uniform cost of PSIP. In this case, the interdiction strategy can be modeled by the polyhedral set $\{z: \sum\limits_{i\in V} k_iz_{i}\leq k, 0\leq z_i\leq 1\}$. Interestingly, PSIP becomes $N\!P$-hard, despite the continuous nature of the interdiction variables. This result is somewhat counterintuitive, as the variables are not binary, and the reduction from the knapsack problem is not straightforward. In the following, we present a proof of this hardness result. 
\begin{theorem}\label{co_nu_psip}
	Continuous PSIP with non-uniform interdiction costs is $N\!P$-hard.
\end{theorem}
\begin{proof}
	We will now show the  hardness result using a reduction from the problem \textsc{Partition}. An instance of \textsc{Partition}  problem is given by natural numbers $c_1,c_2,\dots, c_n$ with $\sum c_j =2b $. The question is to determine whether there exists a set $ S\subseteq [n]$ with value $\sum_{j \in S} c_j = b$. The \textsc{Partition} problem is known to be $N\!P$-hard \citep{garey1979computers}. 
	
	We now construct an instance of PSIP with continuous interdiction and non-uniform interdiction cost. For each number $c_j$ with $j \in [n]$, we introduce a corresponding node $v_j$. The activity nodes can be written as $V=\{0\}\cup[n+1]\cup\{v_{i1}, v_{i2}:\; i\in[n]\}$. The default duration for each activity is $0$. The precedence relationships among activities are depicted in \autoref{fig:reduction_2partition}. There are two parallel paths connecting node $j-1 $ and node $j$, which are joined by $v_{j1}$ and $v_{j2}$, respectively. For any $j\in[n]$, the original duration of  activity $v_{j2} $ is $c_j$ and of $v_{j1} $ is 0. The possible delay of each activity $i$ is defined as $\Delta_i$, where $\Delta_{v_{i2}} = 0$ and $\Delta_{v_{i1}} = 2c_i$. 
    We further let $k_{v_{j1}} = c_j$ and $k_{v_{j2}} = 0$ for $j \in [n]$, and set $k = b$.
    That is, a solution to the constructed PSIP instance is given by a vector $x \in [0, 1]^V$ with $\sum_{i \in V} x_i \leq b$.
    In the following, to simplify notation, we will represent solutions to this PSIP instance by a vector $x \in [0, 1]^{n}$ instead of $[0, 1]^V$, where $x_i$ denotes the degree of delay for activity $v_{i1}$ (note that these activities are the only ones that can be delayed).
	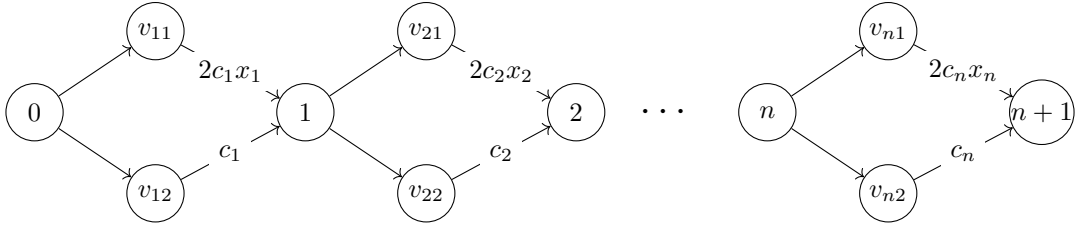
\begin{figure}[h] 
		\centering 
		\begin{tikzpicture}[->,  
			every transition/.style={draw=black},
			every place/.style={draw=black, inner sep=0}
			]
			\node[place] (s0) {$0$};
			\node[place, above right=1.5em and 3em of s0] (s11) {$v_{11}$};
			\node[place, below right =1.5em and 3em  of s0] (s12) {$v_{12}$};
			\node[place, right =8em of s0] (s1) {$1$};
			\node[place, above right=1.5em and 3em of s1] (s21) {$v_{21}$};
			\node[place, below right=1.5em and 3em of s1] (s22) {$v_{22}$};
			\node[place, right =8em of s1] (s2) {$2$};
			\node[place, right =5em of s2] (s3) {$n$};
			\node[place, above right=1.5em and 3em of s3] (s31) {$v_{n1}$};
			\node[place, below right=1.5em and 3em of s3] (s32) {$v_{n2}$};
			\node[place, right =8em of s3] (s4) {$n+1$};
			\path (s2) --++ (2,0) node[midway,scale=1.5] {$\cdots$};
			\draw 
			(s0) edge (s11)
			(s0) edge  (s12)
			(s11) edge node[pos=0.5, fill=white]{$2c_1x_1$} (s1)
			(s12) edge node[pos=0.5, fill=white]{$c_1$} (s1)
			(s1) edge (s21)
			(s1) edge (s22)
			(s21) edge node[pos=0.5, fill=white]{$2c_2x_2$} (s2)
			(s22) edge node[pos=0.5, fill=white]{$c_2$} (s2)
			(s3) edge (s31)
			(s3) edge  (s32)
			(s31) edge node[pos=0.5, fill=white]{$2c_nx_n$} (s4)
			(s32) edge node[pos=0.5, fill=white]{$c_n$} (s4);
			;
		\end{tikzpicture}
		\caption{{Construction of the reduction from \textsc{Partition} to PSIP}\label{fig:reduction_2partition}}
	\end{figure}

    Let $\ell^*(x)$ denote the length of the longest path in the constructed instance under the solution $x \in [0, 1]^n$. Each $0$-$(n+1)$-path contains either $v_{i1}$ or $v_{i2}$, hence
	\begin{align}
	    \ell^*(x)=\max\{2c_1x_1, c_1\}+\cdots+\max\{2c_nx_n, c_n\}=2b+\sum_{i\in [n]:x_i>0.5}2c_ix_i-c_i.
        \label{eq:opt-value-reduction}
	\end{align}
	Next, we show that if the optimal solution value of the PSIP instance is $3b$, then there also is a feasible solution to the \textsc{Partition} instance. 
    Consider an optimal extreme point solution $x^*$ to the PSIP instance with value $\ell^*(x^*) = 3b$.
	By Observation~\ref{obs:extreme-point}, there exists an  $r \in [n]$ such that $x^*_i \in \{0, 1\}$ for all $i \in [n] \setminus \{r\}$. 
    Let $S^*=\{i \in [n] \setminus \{r\}:~x^*_i=1\}$.
    If $x^*_r \leq 0.5$, then from \eqref{eq:opt-value-reduction} we obtain $3b = \ell^*(x^*) = 2b + \sum_{i \in S^*} c_i$ and hence $S^*$ is a feasible solution to the \textsc{Partition} instance.
    If $x^*_r > 0.5$, then from \eqref{eq:opt-value-reduction} we obtain
    $$3b = \ell^*(x^*)=2b+\sum_{i\in [n]:x^*_i>0.5}2c_ix^*_i-c_i = 2b+\sum_{i\in S^*}c_i + c_r(2x^*_r-1) \leq 3b,$$
    which is only possible if $x^*_r = 1$.
    Thus $S^* \cup\{r\}$ is a feasible solution to the \textsc{Partition} instance in this case.
    
	Conversely, if there exists a set $ S $ such that $\sum_{j \in S} c_j=b$, let $x^S$ denote the corresponding solution defined by $x^S_i = 1$ if $i\in S$ and $x=0$ otherwise. Note that $x^S$ satisfies the constraint  $\sum_{i\in[n]} c_ix_i\leq b$ and $\ell^*(x^S)=2b+\sum_{i\in S}2c_ix_i-c_i=3b$.
    This shows that the answer to the given instance of the \textsc{Partition} problem is yes if and only if the maximum longest-path length in our constructed instance is at least $3b$. 
\end{proof}

\begin{remark}
The construction in the proof of Theorem~\ref{co_nu_psip} admits
an integral optimal solution; hence the hardness also holds when the interdiction decisions are required to be binary.
\end{remark}

Following the proof of Theorem~\ref{co_nu_psip}, the continuous interdiction policy is equivalent to the solution to problem~\ref{PSIPDG}, where the uncertainty set is defined as follows:  $$Q_r=\{(\Delta_1,\Delta_2,\dots,\Delta_n):~\sum_{i=1}^n \Delta_i \leq 2b, ~\Delta_i\in [0,2c_i] \text{ for any } i\in [n]\}.$$ Hence, Theorem~\ref{co_nu_psip} directly implies the $N\!P$-hardness of PSIP-DG$_{\mathrm{P}}$.
\begin{corollary}
	PSIP-DG$_{\mathrm{P}}$ is $N\!P$-hard even when there is only one delay group.
\end{corollary}
\subsection{\texorpdfstring{$N\!P$-hardness of \textsc{PSIP-DG}$_{\mathrm{B}}$}{NP-hardness of PSIP-DG-B}}

In this subsection, we prove the $N\!P$-hardness of \textsc{PSIP-DG}$_{\mathrm{B}, \infty}$ and \textsc{PSIP-DG}$_{\mathrm{B}, 1}$, respectively. Moreover, the reduction to \textsc{PSIP-DG}$_{\mathrm{B}, \infty}$ implies an inapproximability result.
	
	\begin{theorem}
		\textsc{PSIP-DG}$_{\mathrm{B}, \infty}$ is NP-complete.
	\end{theorem}
	\begin{proof}
		Given an integer $c$, the decision problem consists of answering if there exists a solution to \ref{dpu} such that maximum delayed project makespan is at least $c$. We will now show that the \textsc{Max $ k $-cover}  problem can be polynomially reduced to this problem.
		
		An instance of \textsc{Max $ k $-cover} is given by a ground set $ U = \{u_1, u_2,\dots,u_n\} $, a collection of $m$
		subsets $S_r \subseteq U, r\in[m]$ of those elements, and an integer $ k $. The task is to select $ k $ subsets such that their union contains as many points as possible. The decision problem asks whether there exists a collection $ I\subseteq \{1,2,\dots,m\} $ such that $ |I|\leq k $ and $ \cup_{r\in I} S_r=U$. The \textsc{Max $ k $-cover} problem is $N\!P$-hard \citep{garey1979computers}. 
		
		Given such an instance of the  \textsc{Max $ k $-cover} problem, we will construct an instance of \textsc{PSIP-DG}$_{\mathrm{B}}$ with $\ell_r\geq |V_r|$. 
		For every element  $ u\in U$, we introduce $m$ parallel activity nodes, $A_u=\{u^1,u^2,\dots, u^m\}$. The activity nodes can be written as $V=\{0, n+1\}\cup\{A_u: u\in U\}$. The default duration for each activity is $0$. Activities in $A_{u_{i+1}}$ preceded activities in $ A_{u_{i}} $, for $ i=1,\dots, n-1$. We connect $u_i^r$ to every node in $A_{u_{i+1}}$, for any $ r \in [m]$.
		
		For each subset $S_r$, we introduce a corresponding delay group $ V_r=\{u_i^r: u_i\in S_r\}$.  Let $ \Delta_i=1$ for $u_i\in S_r, r\in [m]$.  Since $ \ell_r\geq|V_r|$, all the activities' durations in one delay group can be increased to $ 1 $. The interdiction budget is $ k $ and let $c = n$. The complete construction is depicted in \autoref{fig:reduction_MC}.
		\begin{figure}[h]
			\centering
			\begin{tikzpicture}
				\node[circle, draw=black, inner sep=0, minimum size=6mm] (0) at (0, 0) {$0$};
				\foreach \n in {1, ..., 3}{%
					\node[circle, draw=black, right=2em of \the\numexpr\n-1\relax, inner sep=0, minimum size=6mm] (\n) {$u_{\n}^2$};
				}
				\node[circle, draw=black, right=4em of 3, inner sep=0, minimum size=6mm] (4) {};
				\node[circle, draw=black, right=2em of 4, inner sep=0, minimum size=6mm] (5) {$u_{n}^2$};
				\node[circle, draw=black, right=2em of 5, inner sep=-0.5, minimum size=2mm] (6) {$n+1$};
				\draw[->] (0) edge (1) (1) edge (2)  (2) edge (3)  (4) edge (5)  (5) edge (6) ;
				\path (3) --++ (1.5,0) node[midway,scale=1.5] {$\cdots$};
				\foreach \n in {1, ..., 3}{%
					\node[circle, draw=black, above =2em of \n, inner sep=0, minimum size=6mm] (\n1) {$u_{\n}^1$};
					\node[circle, draw=black, below =3em of \n, inner sep=0, minimum size=6mm] (\n3) {$u_{\n}^m$};}
				\node[circle, draw=black, above =2em of 4, inner sep=0, minimum size=6mm] (41) {};
				\node[circle, draw=black, below =3em of 4, inner sep=0, minimum size=6mm] (43) {};
				\node[circle, draw=black, above =2em of 5, inner sep=0, minimum size=6mm] (51) {$u_{n}^1$};
				\node[circle, draw=black, below =3em of 5, inner sep=0, minimum size=6mm] (53) {$u_{n}^m$};
				\foreach \n in {1, ..., 5}{%
					\path (\n) --++ (0,-1) node[midway,scale=1.5] {$\vdots$};
				}
				\draw[->] (0) edge (11) (0) edge (13);
				\draw[->] (51) edge (6) (53) edge (6);
				\path (31) --++ (1.5,0) node[midway,scale=1.5] {$\cdots$};
				\path (33) --++ (1.5,0) node[midway,scale=1.5] {$\cdots$};
				\foreach \n in {1,2,4}{%
					\foreach \m in {1,3}{%
						\draw[->] (\n\m) -- (\the\numexpr\n+1\relax 1);
						\draw[->] (\n\m) -- (\the\numexpr\n+1\relax 3);
						\draw[->] (\n\m) -- (\the\numexpr\n+1);
					}
				}
				\foreach \n in {1,2,4}{%
					\draw[->] (\n) -- (\the\numexpr\n+1\relax 1);
					\draw[->] (\n) -- (\the\numexpr\n+1\relax 3);
				}
			\end{tikzpicture}
			\caption{Construction of the reduction from \textsc{Max $ k $-cover} to \textsc{PSIP-DG} \label{fig:reduction_MC}}
		\end{figure}
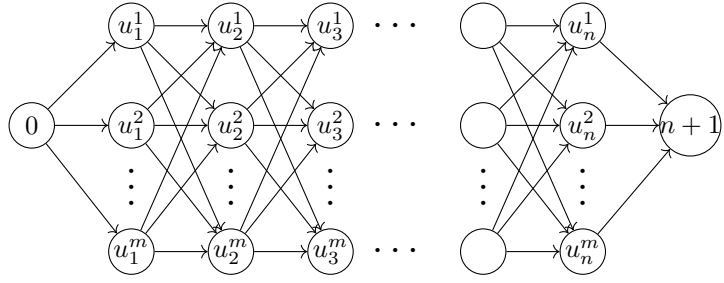
		
		This construction can be done in time that is polynomial in the size of the \textsc{Max $ k $-cover} instance. To prove that this answer is correct, we simply need to show that the original instance of \textsc{Max $ k $-cover} is a yes instance if and only if the created instance is also a yes instance. Suppose that the maximum delayed  makespan is $ n $. Let $ \mathcal{V} $ be the collection of activity groups delayed by the interdictor. By our construction, for each $ V\in \mathcal{V} $, it corresponds to a subset denoted by $S_V$. Let $I_\mathcal{V}=\{S_V: V\in \mathcal{V}\} $ be the corresponding collection of subsets, and it is trivial that $|I_\mathcal{V}|\leq k $.  Consider any element $ u_i $ of $ U $, it corresponds to activity set $A_{u_i}$. Since the delayed project makespan is $ n $, for any $i\in [n]$, there exists $V'\in \mathcal{V}  $ such that $\{u_i^r: r\in [m]\}\cap V' \neq \emptyset $. Therefore, $u_i\in S_{V'} $, which means $I_\mathcal{V} $ is a set cover of $U$.
		
		Now suppose that $U$ has a set cover $I$ of size $ k $. For each subset $S\in I$, let $V_S$ be the associated delay group and the delay groups selected by the interdictor are $ \mathcal{V}_I=\{V_S: S\in I\} $. We have $|\mathcal{V}_I|=|I|= k $. Furthermore, for any $i\in [n]$,  there exists $ S'\in I $ such that $ u_i\in S' $, and thus, $A_{u_i}\cap V_{S'} \neq \emptyset$. The interdictor can increase all the activities' durations to 1 within $V$, then $\sum\limits_{e=(i,j)\in E}y_e\left(\sum\limits_{V\in V_I: i\in V}1\right)=n.$ We conclude that this problem is NP-Complete.
	\end{proof}

\citet{feige1998threshold} showed that \textsc{Max} $ k $-cover cannot be approximated in polynomial time to within a factor of $(1-1/\mathrm{e}+\epsilon)$ unless P=NP, where $\mathrm{e}$ is Euler's number and $\epsilon$ is an arbitrarily small positive number. Our reduction further implies that, unless P = NP, there exists no polynomial-time approximation algorithm to approximate \textsc{PSIP-DG}$_{\mathrm{B}, \infty}$ within a factor smaller than $1-1/\mathrm{e}+\epsilon$.

\begin{corollary}
	It is $N\!P$-hard to approximate \textsc{PSIP-DG} within  $1-1/\mathrm{e}+\epsilon$.
\end{corollary}
Even when we fix $\ell_i=1$, \textsc{PSIP-DG}$_{\mathrm{B}}$ remains NP-complete. The proof is provided as following. 

\begin{theorem}
\textsc{PSIP-DG}$_{\mathrm{B}, 1}$ is NP-complete.
\end{theorem}
\begin{proof}
The reduction is performed from the strongly $N\!P$-hard 3-SAT problem. Given $n$ Boolean variables $x_1, \dots, x_n$ and $q$ clauses $C_1, \dots, C_q$, with each clause consisting of the disjunction of three literals of the variables. Let $L:= \{x_j, \neg x_j : j \in \{1, \dots, n\}\}$ be the set of literals on the variables. We identify each clause with the set of its three literals (e.g., $C_5 = \{x_4, \neg x_3, x_7\}$). A \emph{truth assignment} is a subset $A \subseteq L$ of the literals containing exactly one literal for each variable, i.e., $|\{x_j, \neg x_j\} \cap A| = 1$ for each $j \in  \{1, \dots, n\}$. A clause $C_i$ is fulfilled by truth assignment $A$ if $C_i \cap A \neq \emptyset$. The task is to decide whether there exists a truth assignment that fulfills all clauses. Let $ m_q $ be the number of times the most frequent literals occur in the clauses.

 Given such an instance of \textsc{3-SAT}, we construct a project network with corresponding delay groups and interdiction budget. For each literal $ \ell_{ij}, j\in[3] $ in $ C_i $, there exists an associated \emph{literal gadget} consisting of a directed path with trace $(c_{ij}^1, c_{ij}^2, \dots, c_{ij}^{m_q})$. There exists an activity set $\{c_i: i\in [n]\}$ that corresponds one to one with the clauses. Activity $c_{i}$ is succeeded by $c_{i1}^1, c_{i2}^1, c_{i3}^1$ in parallel. Arcs exist from $c_{ij}^{m_q}$ to $c_{i+1}, ~j=1,2,3$. The construction can be seen in \autoref{fig:reduction_SAT}. The collection of delay groups is then defined as follows:
\begin{itemize}
  \item For every pair of clause positions $(i_1,j_1)$ and $(i_2,j_2)$
        such that $\ell_{i_1 j_1}$ and $\ell_{i_2 j_2}$ are complementary
        literals, we create a two-element delay group
        $\{c_{i_1 j_1}^l, c_{i_2 j_2}^l\}$,where $l \in [m_q]$ is the smallest index such that neither $c_{i_1 j_1}^l$ nor $c_{i_2 j_2}^l$ has yet been assigned to any group.
  \item If, for some $(i,j)$, fewer than $m_q$ copies $c_{ij}^l$ have been
        used in such pairs, the remaining copies are placed in singleton
        groups $\{c_{ij}^l\}$ so that each literal has in total exactly
        $m_q$ associated activity copies.
\end{itemize}

	\begin{figure}[h] 
		\centering 
		\begin{tikzpicture}[yscale=-1.6,xscale=1.5,thick,
			every transition/.style={draw=red,fill=red!20,minimum size=4.8mm},
			every place/.style={draw=black,minimum size=6mm}]
			\node[place,label=above:{}] (c1) at (-4,0) {$c_1$};
			\node[place,label=above:{}] (c2) at (0.2,0) {$c_2$};
			\node[place,label=above:{}] (c3) at (1.5,0) {$c_{q}$};
			\node[place,label=above:{}] (c4) at (5.7,0) {$c_{q+1}$};
			\node[place,label=above:{}] (l1) at (-3,-1) {$c^1_{11}$};
			\node[place,label=above:{}] (l4) at (2.5,-1) {$c^1_{q1}$};
			\node[place,label=above:{}] (l11) at (-2,-1) {$c^2_{11}$};
			\node[place,label=above:{}] (l41) at (3.5,-1) {$c^2_{q1}$};
			\node[place,label=above:{}] (l12) at (-0.8,-1) {$c^{m_q}_{11}$};
			\node[place,label=above:{}] (l42) at (4.7,-1) {$c^{m_q}_{q1}$};
			
			\node[place] (l2) at (-3,0) {$c_{12}^1$};
			\node[place,label=above:{}] (l21) at (-2,0) {$c^2_{12}$};
			\node[place,label=above:{}] (l22) at (-0.8,0) {$c^{m_q}_{12}$};
			\node[place] (l3) at (-3,1) {$c_{13}^1$};		
			\draw (c1) edge  (l1) edge [post]  (l2) edge [post]  (l3); 
			\node[place] (l31) at (-2,1) {$c^2_{13}$};
			\node[place,label=above:{}] (l32) at (-0.8,1) {$c^{m_q}_{13}$};
			\draw (l1) edge[post] (l11) (l2)edge [post]  (l21)  (l3)edge [post]  (l31); 
			\path (l11) --++ (1,0) node[midway,scale=1.5] {$\cdots$};
			\path (l21) --++ (1,0) node[midway,scale=1.5] {$\cdots$};
			\path (l31) --++ (1,0) node[midway,scale=1.5] {$\cdots$};
			\draw (c2) edge[pre] (l12) edge [pre]  (l22)  edge [pre]  (l32); 
			\path (c2) --++ (1,0) node[midway,scale=1.5] {$\cdots$};
			
			\node[place] (l5) at (2.5,0) {$c_{q2}^1$};
			\node[place,label=above:{}] (l51) at (3.5,0) {$c^2_{q2}$};
			\node[place,label=above:{}] (l52) at (4.7,0) {$c^{m_q}_{q2}$};
			\node[place] (l6) at (2.5,1) {$c_{q3}^1$};		
			\draw (c3) edge  (l4) edge [post]  (l5) edge [post]  (l6); 
			\node[place] (l61) at (3.5,1) {$c^2_{q3}$};
			\node[place,label=above:{}] (l62) at (4.7,1) {$c^{m_q}_{q3}$};
			\draw (l4) edge[post] (l41) (l5)edge [post]  (l51)  (l6)edge [post]  (l61); 
			\path (l41) --++ (1,0) node[midway,scale=1.5] {$\cdots$};
			\path (l51) --++ (1,0) node[midway,scale=1.5] {$\cdots$};
			\path (l61) --++ (1,0) node[midway,scale=1.5] {$\cdots$};
			\draw (c4) edge[pre] (l42) edge [pre]  (l52)  edge [pre]  (l62); 
		\end{tikzpicture}
		\caption{{Construction of the reduction from 3-SAT to \textsc{PSIP-DG}$_{\mathrm{B}, 1}$}\label{fig:reduction_SAT}} 
	\end{figure}
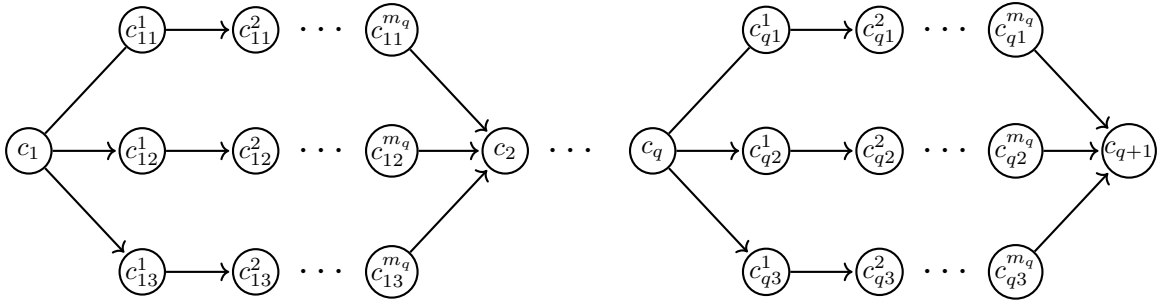 
The collection of delay groups $ \mathcal{V} $ is constructed as described above, where $m=|\mathcal{V}|$. Let $d_v=0$ for all $ v\in V$ and $ \Delta_r(v)=1$ for $v\in V_r$ and $r\in [m]$. The interdiction budget is set as $ k=q\cdot m_q $, which completes the construction of the problem instance. Next, we show that the maximum interdicted project makespan is $ q\cdot m_q $ if and only if there exists a satisfying truth assignment for the \textsc{3-SAT} instance. 

First, assume there exists a truth assignment $A$ fulfilling all clauses. For each $ i\in [q] $, there exists $ i(A)\in [3] $ such that $\ell_{ii(A)} \in A\cap C_i$, let the interdictor delay the activities $ c^1_{ii(A)}, c^2_{ii(A)},\dots, c^{m_q}_{ii(A)}$. According to the construction of the delay groups, there must be a delay group containing $ c^r_{ii(A)}, r\in [m_q] $, denoted by $ V^r_{ii(A)} $. It is obvious that $ V^{r_1}_{{i_1}{i_1}(A)}\cap V^{r_2}_{{i_2}{i_2}(A)}=\emptyset$ for any $ r_1, r_2\in [m_q] $ and $ i_1, i_2\in [q]$; otherwise $\ell(c_{i_1i_1(A)})=\neg\ell(c_{i_2i_2(A)})$, which contradicts the satisfying condition. Thus the interdictor can select $ q\cdot m_q $ delay groups $\{ V^r_{ii(A)}\}_{i\in [q], r\in [m_q]}$ by increasing the duration of activity $c^r_{ii(A)}$ to $ 1 $ in each attacked group $V^r_{ii(A)}$, which leads to a total project completion time of $ q\cdot m_q $. 

If the maximum project makespan equals $ q\cdot m_q $, it indicates that all activities on one critical path are interdicted. For $i\in [q]$, assume that on the critical path from $ c_i $ to $ c_{i+1} $, all the activities in the \emph{literal gadget} of $\ell_{ij(i)}$ are delayed, that are  $ \{c^r_{ij(i)}\}_{r\in[m_q]} $. We define $ A=\{\ell_{ij(i)}:~ i\in[q]\} $ as the truth assignment.  We will prove that for any $ i_1 \neq i_2 $, $ \ell(i_1j_{i_1}) \neq \neg\ell(i_2j_{i_2}) $ holds.

Suppose there exists $ i_1 \neq i_2 $ such that $ \ell(i_1j_{i_1})=\neg\ell(i_2j_{i_2})$. Then within the delay groups, there must exist some $ \text{index}_1, \text{index}_2\in [m_q] $ such that $\{c^{\text{index}_1}_{i_1j_1},c^ {\text{index}_2}_{i_2j_2}\}\in \mathcal{V}$. Given the constraint $ \ell_i=1 $, exactly one of the activities $ c^{\text{index}_1}_{i_1j_1} $  and $ c^ {\text{index}_2}_{i_2j_2}  $ can be interdicted. So the project makespan will be at most $ q\cdot m_q-1 $, which contradicts the assumption. Therefore, $ A\cap C_i =\ell(ij_i)$ and the assignments in set $ A $ do not conflict. Thus, the assignment $ A $ satisfies all clauses.
\end{proof}
Each delay group in the \textsc{PSIP-DG} instance constructed above consists of at most 2 activities. Hence, this problem remains $N\!P$-hard even when $|V_r|$ is relatively small but greater than 1. Another important observation from this reduction is that even if we set the interdiction budget as a large number, e.g., $k\geq q\cdot m_q$, \textsc{PSIP-DG}$_{\mathrm{B}, 1}$ remains $N\!P$-hard. 

\section{\texorpdfstring{Heuristic algorithms for PSIP-DG\_${\mathrm{B}}$}{Heuristic algorithms for PSIP-DG\_B}}  
\label{sec:heuristics}

The bilinear structure of DP($U^B$), which couples path variables and interdiction variables, suggests two complementary search spaces. SRS searches over paths, computing for each candidate path the exact optimal interdiction restricted to it; GSS searches over interdiction policies, computing for each candidate policy the exact longest path. Each heuristic thus performs local search in one block of variables while solving the other block exactly.
\subsection{Approximation algorithm}
\label{sec:approximation}
In this section, we focus on the design of algorithms to solve \textsc{PSIP-DG}$_{\mathrm{B}}$. In this case, the impact of delaying each single group, denoted as $ W_r $,  can be computed efficiently in polynomial time. A natural way to quickly find solutions to  \textsc{PSIP-DG}$_{\mathrm{B}}$ is to start from the empty attacked group set and to add elements one at a time, taking at each step that element that increases a longest path the most. The resulting  procedure, which we now define formally, is called the \emph{greedy heuristic}.
	\begin{breakablealgorithm}
	\caption{The greedy heuristic for \textsc{PSIP-DG}$_{\mathrm{B}}$}\label{greedy_heuristic}
	\begin{algorithmic}[]
		\State \textbf{Initialize}: $\mathcal{M}=[m]$, $ d_i^0=d_i$ for all $ i\in [m]$,  ${d}^0=\{d_i^0:~i\in [m] \}$, $ \mathcal F=\emptyset.$ 
		\For{$ t \gets 0$ to $k-1$} 
		\State \textbf{Step 1}: For each $ r\in \mathcal{M}$, solve \ref{dplr} w.r.t $d^0$ and calculate the optimal solution $ x_{r}^{t} $ and the maximum delay $ W_{r} $ of attacking the delay group $ V_r$.
	    \State \textbf{Step 2}: Schedule the delay groups 
	    in non-increasing attack-profit order with ties settled arbitrarily, such that $W_{r^{t}_1}\geq W_{r^{t}_2}\geq \dots \geq W_{r^{t}_{|\mathcal{M}|}}$. Consume one unit of the interdiction budget in the first group $ r^{t}_1 $ and delay the activities in the selected group $ r^{t}_1 $ according to the optimal solution $ x^{t}_{r^{t}_1}$. 
	    \State \textbf{Step 3}:  Update $\mathcal F:=\mathcal F\cup \{(r^{t}_1, x^{t}_{r^{t}_1})\}$, $\mathcal{M}=\mathcal{M}\backslash\{r^{t}_1\}$, $ {d}^{t+1}:= {d}^{t}+x^{t} $, and  $ t:=t+1 $.
		\EndFor
		\State \textbf{Output}: $ \mathcal F $
	\end{algorithmic}

\end{breakablealgorithm}

Let $ C_0 $ be the optimal makespan without disruption and let $C^*$ be the value of an optimal solution to problem \ref{dplr}. We will now show that the constructed algorithm is a $k$-approximation algorithm.

\begin{theorem}
	Algorithm \ref{greedy_heuristic} gives a solution with approximation guarantee of $ k $.
\end{theorem}
\begin{proof}
	We can derive from Step $ 2 $ of Algorithm \ref{greedy_heuristic} that  $ C^*-C_0\leq k(W_{{r^{0}_1}}-C_0)$. Suppose $ C^*-C_0> k(W_{{r^{0}_1}}-C_0)$ and let $ P^* $ be the critical path that determines the makespan. There exists at least one delay group $ V_{r^*} $ such that $ W_{{r^{*}}}-C_0>W_{{r^{0}_1}}-C_0.$ It contradicts the fact that delaying the delay group $ V_{r^{0}_1} $ in the first iteration has the most impact on the longest path. Hence $ C^*\leq kW_{{r^{0}_1}} $ and it gives a   $k$-approximation.
\end{proof}

\subsection[Heuristics]{Two structure-based heuristics: searching over paths and over policies}
\label{sec:localsearch}

We propose two heuristics for \textsc{PSIP-DG}$_{\mathrm{B}}$, each built around one of the two blocks of variables in the bilinear formulation \eqref{dpu}: the choice of a path and the choice of an interdiction policy. The \emph{Subpath Reoptimization Search} (SRS)  iteratively improves a candidate path and for each candidate solves the restricted interdiction problem exactly by dynamic programming. The \emph{Group Selection Search} (GSS) iteratively improves a candidate interdiction policy, exchanging selected groups and evaluating each candidate policy exactly via a longest-path computation. Both heuristics use the tractable substructures of PSIP-DG discussed in Section ~\ref{sec:tractable_cases}: SRS builds its initial path from a longest-path relaxation solved by dynamic programming, while GSS initializes from the exact per-group decomposition, selecting the $k$ groups with the largest individual contributions $W_r$. As Observation~\ref{obs:interaction} shows, these individual contributions cannot capture the interaction of groups along
shared paths; the search phases of both heuristics are designed precisely to recover this interaction. 

\begin{observation}\label{obs:interaction}
For $S \subseteq [m]$, let $f(S)$ denote the increase in the worst-case makespan when the groups in $S$ are attacked, and let $W_r = f(\{r\})$ denote the individual contribution of group $V_r$. In general, $f(S)$ is
not determined by the individual contributions $\{W_r\}_{r \in S}$, and $f$ is neither additive nor submodular. In the instance of Figure~\ref{fig:illustrative_case}, all three groups have identical individual contributions, $W_1 = W_2 = W_3 = 2$, yet $f(\{2,3\}) = 5 > W_2 + W_3$, while $f(\{1,3\}) = 4$.
\end{observation}
We present SRS and GSS in the following subsections.

\subsubsection{Subpath reoptimization search (SRS)} 
The SRS heuristic (Algorithm~\ref{localsearch1}) first identifies an initial critical path in the network by dynamic programming, accounting for the possible delays under budgeted interdiction where applicable. For the case of \textsc{PSIP-DG}$_{\mathrm{B},1}$, we relax the constraint on the activity groups and assume that $k$ activities can be attacked. This relaxation transforms \textsc{PSIP-DG}$_{\mathrm{B}}$ to PSIP, allowing the maximum delayed length to be solved by dynamic programming. For the case of \textsc{PSIP-DG}$_{\mathrm{B},\infty}$, the initial path is instead the nominal longest path, onto which the group delays are subsequently added. This yields an initial path that approximates the most delay-sensitive structure in the network. Given the initial path, the algorithm determines the selection of delay groups that contributes the most to increasing the length of the path. This forms the initial solution. The detailed pseudocode of the underlying subroutines is provided in Algorithms~\ref{alg:find_max_delay_path}, \ref{alg:dp_group_B1} and~\ref{alg:dp_group_Binf} in the appendix. 

To explore the neighborhood of a given solution path, we iteratively modify a fixed-length subpath by replacing it with an alternative segment and evaluating the resulting delay. Specifically, for each interval of a predetermined number of edges on the current path, we reconnect the start and end nodes of the interval using a new subpath computed via dynamic programming. The budget is split between the two parts: for each $k' \le k$, the
algorithm allocates $k'$ units of the budget to the subpath and the remaining $k-k'$ units to the rest of the path, computing the maximum delayed subpath under this relaxation (the group structure is dropped, and for \textsc{PSIP-DG}$_{\mathrm{B},\infty}$ the per-group limits $\ell_r$ are flattened to a common bound $\bar\ell = \max_{r \in [m]} \ell_r$). The maximum delay of the residual path segment is approximated by applying the residual budget greedily. Every alternative whose approximate length is at least that of the current path is retained for exploration. The approximation serves only as a filter: retained neighbors are explored in breadth-first order, and the objective of each is
recomputed exactly by the dynamic program used in the initialization. The search stops once the neighbor limit $M$ is reached or no further
improvement is found. This procedure is depicted in Algorithm~\ref{alg:get_path_neighborhood}. In our experiments, $k$ and $M$ are fixed to  constants, so the practical running time is $\mathcal{O}\big(n^{2} + n|E| + n^{2}\log n\big).$

\begin{breakablealgorithm}  
\caption{Subpath reoptimization search (SRS)}
\label{localsearch1}
\begin{algorithmic}[1]
\State \textbf{Initialization:}
\If{PSIP-DG$_{\mathrm{B},1}$}
  \State $p^0 \gets \texttt{find\_max\_delay\_path}(G,s,t,k)$   \Comment{initial path, Alg.~A.1}
\Else
  \State compute a longest path $p^0$ in $G$
\EndIf
\State $(\text{delay}^0, S^0) \gets \texttt{DP\_on\_path}(p^0, k)$   \Comment{exact evaluation, Alg.~A.2/A.3}
\State $\text{best\_length} \gets \sum_{i\in p^0} d_i + \text{delay}^0$,\; $\text{best\_policy} \gets S^0$
\State $Q \gets \{p^0\}$,\quad $V \gets \emptyset$        \Comment{$V$: explored paths}
\While{$Q \neq \emptyset$ \textbf{and} $|V| < M$}
  \State select and remove a path $p$ from $Q$
  \State $(\mathcal N_p, l_p, S_p) \gets \texttt{get\_path\_neighborhood}(p, k, \texttt{num})$   \Comment{Alg.~A.4}
  \If{$l_p > \text{best\_length}$}
    \State $\text{best\_length} \gets l_p$,\; $\text{best\_policy} \gets S_p$
  \EndIf
  \State $V \gets V \cup \{p\}$
  \For{each $p' \in \mathcal N_p$ with $p' \notin V$}
    \State add $p'$ to $Q$
  \EndFor
\EndWhile
\State \textbf{Return:} $\text{best\_policy}, \text{best\_length}$

\end{algorithmic}  
\end{breakablealgorithm}

\subsubsection{Group Selection Search}

The GSS heuristic, as shown in Algorithm~\ref{localsearch2}, builds on the per-group decomposition of Section~\ref{sec:tractable_cases}: the maximum delay impact $W_r$ of each group $V_r$, $r \in [m]$, is
computed by solving problem \ref{dplr}, and the initial solution consists of the $k$ groups with the largest impacts. Based on this idea, we evaluate the optimal delay achievable for each group individually. The initial solution consists of the $k$ most impactful groups, providing a high-quality starting point. Neighborhood exploration iteratively perturbs the current solution: for \textsc{PSIP-DG}$_{\mathrm{B},1}$, the delayed activity of a selected group is exchanged with an activity drawn at random from a non-selected
group; for \textsc{PSIP-DG}$_{\mathrm{B},\infty}$, a selected group is swapped with a non-selected one. Each neighbor is evaluated by the induced longest-path length. A candidate whose length is at least that of the current solution is accepted and becomes the basis of the next round, while the best solution encountered is recorded separately. The process continues until $M$ distinct candidates have been evaluated or no new candidate can be generated.  The overall running time of GSS is $\mathcal{O}\big((m + M)|E| + m \log m\big)$.

\begin{breakablealgorithm}
\caption{Group Selection Search (GSS)}
\label{localsearch2}
\begin{algorithmic}[1]
\State \textbf{Initialization:}
\For{each $r \in [m]$}
  \State compute $W_r$ and attacked activities $\mathcal{I}_r^{*}$ by solving subproblem $\textsc{MAX}(\ell, r)$
\EndFor
\State rank groups in non-increasing order of $W_r$
\State select the top $k$ groups $\mathcal{R}_0$ and construct the initial policy
       $\mathcal{I}_0$ from $\{\mathcal{I}_r^{*}\}_{r \in \mathcal{R}_0}$
\State $\mathcal{I} \gets \mathcal{I}_0$,\; $\mathcal{N} \gets \emptyset$
\Comment{current policy}
\State $\text{cur\_length} \gets l(\mathcal{I}_0)$,\;
       $\text{best\_policy} \gets \mathcal{I}_0$,\; $\text{best\_length} \gets l(\mathcal{I}_0)$
\While{$|\mathcal{N}| < M$}
  \State draw $q$ random neighbours $\mathcal{N}(\mathcal{I})$ of $\mathcal{I}$, each exchanging
         $s$ selected items for $s$ non-selected ones:
  \State \quad \textbf{if} \textsc{PSIP-DG}$_{\mathrm{B},1}$: exchange activities, the incoming ones
         drawn from $s$ \emph{distinct} non-selected groups
  \State \quad \textbf{if} \textsc{PSIP-DG}$_{\mathrm{B},\infty}$: exchange entire groups
  \If{$\mathcal{N}(\mathcal{I}) \setminus \mathcal{N} = \emptyset$}\textbf{break}
    \Comment{no new neighbour can be generated}
  \EndIf
  \For{each $\mathcal{I}' \in \mathcal{N}(\mathcal{I}) \setminus \mathcal{N}$}
    \State compute $l(\mathcal{I}')$ exactly by a longest-path evaluation
    \If{$l(\mathcal{I}') \geq \text{cur\_length}$}
      \State $\mathcal{I} \gets \mathcal{I}'$,\; $\text{cur\_length} \gets l(\mathcal{I}')$
      \Comment{equal moves accepted: plateau drift}
    \EndIf
    \If{$l(\mathcal{I}') \geq \text{best\_length}$}
      \State $\text{best\_length} \gets l(\mathcal{I}')$,\; $\text{best\_policy} \gets \mathcal{I}'$
    \EndIf
    \State $\mathcal{N} \gets \mathcal{N} \cup \{\mathcal{I}'\}$
  \EndFor
\EndWhile
\State \textbf{Output:} $\text{best\_policy}$, $\text{best\_length}$
\end{algorithmic}
\end{breakablealgorithm}

\section{Computational study}  
\label{sec:experiment}
In this section, we conduct a computational study to evaluate the performance of SRS and GSS on \textsc{PSIP-DG}$_{\mathrm{B},1}$ and \textsc{PSIP-DG}$_{\mathrm{B},\infty}$. We tested different variants of the two heuristics to determine good parameter settings and combinations. We assess our methods on self-generated networks, comparing against an exact MIP approach solved by Gurobi. We report two metrics: solution quality, measured against the best solution found by Gurobi within a time limit of $3{,}600$ seconds, which is optimal whenever the solver proves optimality, and running time. The exact method is programmed using Gurobi 13.0.1, and the heuristic algorithms are implemented in Python 3.12.2. All experiments are run on an Intel(R) Core(TM) Ultra 9 processor at 2.90 GHz, with 64 GB of RAM.

\subsection {Test instances}
\subsubsection{Project network instances}

Since no benchmark instances exist for our problem, we generate random instances using RanGen \citep{vanhoucke2008evaluation, demeulemeester2003rangen}, which creates activity-on-the-node project networks with controlled characteristics: the parameter $I_1$ sets the number of activities, and $I_2 \in [0,1]$ measures the closeness of the network to a serial or parallel graph. As RanGen supports at most slightly fewer than $1{,}000$ activities, we vary the number of activities from 50 to 900 and, for each size, a subset of $I_2$ values between $0.01$ and $0.8$. The resulting networks range from small to large and from parallel to serial. RanGen also assigns durations to each activity, providing input data for the general project scheduling problems. 

To explore scalability beyond this range, we generate DAG instances with up to $5{,}000$ nodes by a layered procedure: the number of nodes per layer is drawn from a specified range, arcs are created between adjacent layers and within layers according to given probabilities,
and isolated nodes are connected to a random node in the next layer. Finally, the procedure adds a start node connected to every node without a predecessor and an end node preceded by the nodes of the last layer. Each node is then assigned a random integer duration. By varying the arc-creation probabilities and layer widths, we
generate instances at three edge-density levels for each target size: sparse, medium, and dense, with the average number of arcs per node ranging from about $4$ to $40$.
Since the number of nodes per layer is drawn at random, the nominal sizes are approximate; all tables report the realized $|V|$ and $|E|$.

\subsubsection{Random delay groups and interdiction budget}

To complete the input for \textsc{PSIP-DG}$_{\mathrm{B}}$, we randomly partition the activities into $m$ delay groups of balanced or randomly varying sizes.
The delay increments are drawn independently and uniformly at random from the integers in $[0, d_{\max}]$, where $d_{\max} = \max_j d_j$ is the largest nominal duration in the network; the dummy source and sink have zero duration and zero delay. Combined with varying interdiction budgets, this setting yields general-purpose test instances that are not tied to a specific application domain.

We classify the instances into small (50--200 activities), medium (400--600), large (800--900), and extra-large ($>1{,}000$); see \autoref{tab:instances} in Appendix~B for summary statistics. The sets of instances are named $S\_|V|$, where $|V| = I_1$ represents the number of non-dummy activities; each network contains $|V| + 2$ nodes including the dummy source and sink. Each instance is named $|V|\_I_2\_m\_k\_(\text{i}/\text{ii})$, with $m$ the number of delay groups, $k$ the interdiction budget, and i and ii indicating \textsc{PSIP-DG}$_{\mathrm{B},1}$ and \textsc{PSIP-DG}$_{\mathrm{B},\infty}$, respectively. For the extra-large instances, $I_2$ is replaced by the number of edges. 

\subsection{Results}

We now present the results of SRS and GSS on the instance sets described above and compare them with the exact MIP approach. The parameter settings of both heuristics follow a common rule established in preliminary tests and are kept fixed across all experiments: the neighbor limit $M$ of GSS is set to 100 for small, 200 for medium, and 500 for large and extra-large instances, while SRS uses $M = 100$ throughout. Since GSS is non-deterministic, each instance is run ten times and we report the average.

We solve the linearized formulation of \eqref{dpu} (see \ref{appendix: MIP}) with Gurobi under a time limit and denote by $Z_{\mathrm{G}}$ the objective value of the best solution found within the limit; when Gurobi proves optimality, $Z_{\mathrm{G}}$ is the optimal value. Let $\mathrm{ALG}$ denote the objective value of a heuristic and $L^*$ the longest path length in the undisturbed network. We report three metrics: the gap to Gurobi,
$\mathrm{Gap} = (Z_{\mathrm{G}} - \mathrm{ALG})/Z_{\mathrm{G}}$; the delay gap, $\mathrm{DGap} = (Z_{\mathrm{G}} - \mathrm{ALG})/(Z_{\mathrm{G}} - L^*)$, which measures the difference relative to the induced delay only, independently of the undisturbed makespan; and the improvement ratio of the search phase, $\mathrm{IMP} = (\mathrm{ALG} - \mathrm{ALG}^0)/(\mathrm{ALG}^0 - L^*)$, where $\mathrm{ALG}^0$ is the initial solution.

\subsection{Comparison of the algorithms}
\subsubsection[Results for \textsc{PSIP-DG}B1]{Results for \textsc{PSIP-DG}$_{\mathrm{B, 1}}$}

\autoref{tab:AVG_results_i} presents the average results of SRS and GSS, aggregated by network size $|V|$ and averaged over the topological parameter $I_2$, the number of delay groups $m$, and the interdiction budget $k$. It can be seen from the table that SRS outperforms GSS in terms of solution quality across all instance sets, with an average Gap below $0.3\%$ and even the stricter DGap stays below $0.8\%$ on average. The Gap of SRS does not go up with the number of activities  and is as low as $0.13\%$ for the largest sets ($S\_800$ and $S\_900$), where GSS performs clearly worse ($1.07\%$ at $S\_{900}$). Both algorithms are fast in this case: SRS requires at most four seconds on average and GSS less than $0.5$ seconds. However, the two searches behave quite differently. GSS almost always exhausts its neighbor limit ($N \approx M$ in all rows), whereas SRS reaches a better solution after exploring fewer than $20$ neighbors, at a higher cost per neighbor due to the reoptimization of each candidate path. Since Gurobi proves optimality
on every instance in these sets, the gaps are measured against the optimum: SRS finds an optimal solution on $85.4\%$ of the instances and falls within $1\%$ of the optimum on $94.7\%$ of the instances.

\begin{table}[htbp]
	\centering
	\begin{threeparttable}
	\caption{Comparison of SRS and GSS on different instance sets of \textsc{PSIP-DG}$_{\mathrm{B}, 1}$}	\label{tab:AVG_results_i}	
	\begin{tabular}{@{}lcccccccccc@{}}
		\toprule
		{}& \multicolumn{4}{c}{{SRS}} & \multicolumn{4}{c}{{GSS}} \\ 
		\cmidrule(lr){2-5} \cmidrule(lr){6-9}
		{Instances set} & {Gap} & {DGap} & ${T(s)}$ & ${N}$  & {Gap} & {DGap} & ${T(s)} $& ${N}$ \\
		\midrule
        $|V|=50$  & 0.26\% & 0.71\% & 0.00 & 1.28  & 0.48\% & 1.07\% & 0.00 & 89.72  \\
$100$     & 0.20\% & 0.56\% & 0.01 & 2.06  & 1.54\% & 3.49\% & 0.00 & 104.87 \\
$200$     & 0.16\% & 0.50\% & 0.13 & 7.35  & 1.97\% & 5.00\% & 0.01 & 107.03 \\
$400$     & 0.15\% & 0.41\% & 0.50 & 12.58 & 1.09\% & 3.35\% & 0.04 & 206.88 \\
$600$     & 0.24\% & 0.60\% & 2.67 & 19.28 & 2.10\% & 5.79\% & 0.12 & 208.32 \\
$800$     & 0.13\% & 0.35\% & 3.53 & 16.75 & 0.82\% & 2.38\% & 0.24 & 506.76 \\
$900$     & 0.13\% & 0.39\% & 2.23 & 14.84 & 1.07\% & 2.88\% & 0.43 & 507.19 \\
		\bottomrule
	\end{tabular}
    \begin{tablenotes}
    \small
    \item \emph{Note.} Each row averages over all $(I_2, m, k)$ configurations
    with the given number of nodes $|V|$, using $I_2 \in \{ 0.05, 0.1, 0.5, 0.8\}$.
    The number of configurations per row ranges from $39$ ($S\_50$) to
    $160$ ($S\_900$). $T(s)$: average CPU time in seconds; $N$: average number
    of explored neighbors; $\mathrm{Gap} = (Z_{\mathrm{G}}-\mathrm{ALG})/Z_{\mathrm{G}}$;
    $\mathrm{DGap} = (Z_{\mathrm{G}}-\mathrm{ALG})/(Z_{\mathrm{G}}-L^*)$.
    \end{tablenotes}
\end{threeparttable}
\end{table}

For the extra-large instances of \textsc{PSIP-DG}$_{\mathrm{B},1}$, both heuristics are near-optimal: the initial solutions are already close to the optimum, and the search phase improves them by at most $1\%$. Gurobi solves these instances quickly as well, so the heuristics offer no runtime advantage for this variant; detailed
results are reported in \ref{appendix:extralarge_i}
(\autoref{tab:AVG_results_large_i}). The situation is different for \textsc{PSIP-DG}$_{\mathrm{B},\infty}$, as we show next.

\subsubsection[Results for \textsc{PSIP-DG}BINF]{Results for \textsc{PSIP-DG}$_{\mathrm{B}, \infty}$}

\autoref{tab:AVG_results_ii} shows that GSS clearly outperforms SRS for \textsc{PSIP-DG}$_{\mathrm{B},\infty}$ in both solution quality and running time, and the advantage widens with the instance size. The Gap of GSS stays below $1.4\%$ across all sets, whereas the Gap
of SRS reaches $3.94\%$ at $S\_600$ and remains close to $4\%$ for the largest sets; the DGap shows the same picture, staying below $3.4\%$ for GSS while exceeding $11\%$ for SRS on the three largest sets. GSS is also much faster, requiring at most $0.2$ seconds on average compared to up to $17$ seconds for SRS. Interestingly, unlike in the \textsc{PSIP-DG}$_{\mathrm{B},1}$ case, GSS does not reach its neighbor limit here and terminates well before $N$ reaches $M$.

\begin{table}[htbp]
\centering
\begin{threeparttable}
\caption{Comparison of SRS and GSS on different instance sets of \textsc{PSIP-DG}$_{\mathrm{B},\infty}$}
\label{tab:AVG_results_ii}

\begin{tabular}{@{}lcccccccc@{}}
\toprule
 & \multicolumn{4}{c}{SRS} & \multicolumn{4}{c}{GSS} \\
\cmidrule(lr){2-5} \cmidrule(lr){6-9}
Instances set & Gap & DGap & $T(s)$ & $N$ & Gap & DGap & $T(s)$ & $N$ \\
\midrule
$|V|=50$  & 0.66\% & 1.75\%  & 0.01  & 12.10 & 0.05\% & 0.12\% & 0.00 & 57.64  \\
$100$     & 1.02\% & 2.85\%  & 0.15  & 47.71 & 0.58\% & 1.33\% & 0.00 & 77.31  \\
$200$     & 1.53\% & 3.99\%  & 0.67  & 56.75 & 0.99\% & 2.51\% & 0.01 & 87.32  \\
$400$     & 2.92\% & 8.62\%  & 4.39  & 84.88 & 0.75\% & 2.19\% & 0.02 & 156.04 \\
$600$     & 3.94\% & 11.38\% & 9.19  & 88.42 & 1.31\% & 3.38\% & 0.06 & 175.53 \\
$800$     & 3.58\% & 11.28\% & 16.43 & 95.56 & 0.45\% & 1.25\% & 0.13 & 363.28 \\
$900$     & 3.64\% & 11.38\% & 17.21 & 93.08 & 0.69\% & 1.91\% & 0.19 & 378.81 \\
\bottomrule
\end{tabular}
\begin{tablenotes}
\small
\item \emph{Note.} Same setup and metrics as Table~\ref{tab:AVG_results_i}.
\end{tablenotes}
\end{threeparttable} 
\end{table}

\begin{figure}[h!]
	\centering		\includegraphics[width=\textwidth]{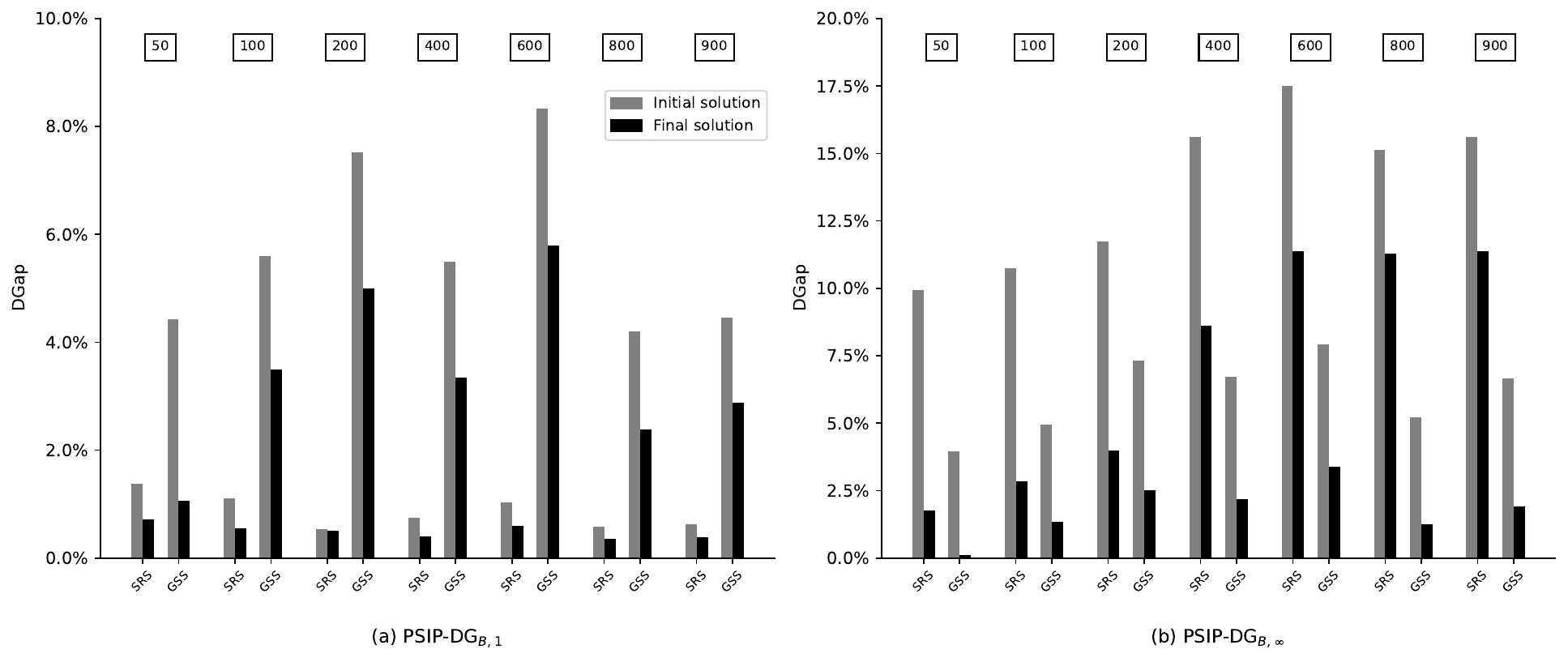}
	
\caption{DGap of the initial (grey) and final (black) solutions of
SRS and GSS for both variants.}\label{fig:ini}
\end{figure}
The results in \autoref{fig:ini} show that the selection of the
initial solution affects the performance of the heuristics and
illustrate the extent to which the search procedures enhance the
quality of the initial solutions. In \textsc{PSIP-DG}$_{\mathrm{B},1}$,
the initial path of SRS is already close to the best solution found:
on S\_900, its average DGap is $0.63\%$ initially and $0.39\%$
after the search, so the search has little room to improve.
Conversely, GSS more substantially improves the initial solution,
reducing the average gap from $4.45\%$ to $2.88\%$ on S\_900\_i, a
relative decrease of about $35\%$. The reason is that SRS builds its
initial path from a tight relaxation of the model (Section~4.2.1);
for the case where $\ell=1$, this relaxation is especially
tight, so the initial path is nearly optimal and the search converges
after exploring only a few neighbors.

On the other hand, in \textsc{PSIP-DG}$_{\mathrm{B},\infty}$ the
situation is reversed: the initial solution of GSS is already good,
whereas that of SRS is far from optimal. The initial DGap of GSS lies between about $4\%$ and $8\%$, while SRS starts with an initial DGap of $10\%$ to $17\%$. Consequently, the search phase of SRS must explore many more neighbors and, since each rerouted subpath is
reoptimized with a larger delay budget under \textsc{PSIP-DG}$_{\mathrm{B},\infty}$, each evaluation is also more expensive, resulting in a considerably longer running time than for
\textsc{PSIP-DG}$_{\mathrm{B},1}$ (\autoref{tab:AVG_results_ii}).

\begin{table}[h!]
\centering
\begin{threeparttable}
\caption{Comparison of GSS and Gurobi on medium, large, and
extra-large instance sets of \textsc{PSIP-DG}$_{\mathrm{B},\infty}$}
\label{tab:superior performance}
\begin{tabular}{@{}lcccccc@{}}
\toprule
 & \multicolumn{4}{c}{GSS} & \multicolumn{2}{c}{GRB} \\
\cmidrule(lr){2-5} \cmidrule(lr){6-7}
Instance set & $T(s)$ & $N$ & Gap & DGap & $T(s)$ & optimal \\
\midrule
S\_400            & 0.02 & 156.04 & 0.75\% & 2.19\% & 1.06   & all   \\
S\_600            & 0.06 & 175.53 & 1.31\% & 3.38\% & 7.32   & all   \\
S\_800            & 0.13 & 363.28 & 0.45\% & 1.25\% & 74.95  & all   \\
S\_900            & 0.19 & 378.81 & 0.69\% & 1.91\% & 110.00 & all   \\
\midrule
$|V|\approx 1200$ & 0.57 & 463.1  & 0.24\% & 0.91\% & 1532   & 12/18 \\
$2000$            & 1.50 & 462.1  & 0.19\% & 0.76\% & 3541   & 0/18  \\
$3000$            & 1.73 & 459.6  & 0.12\% & 0.48\% & 3116   & 2/18  \\
$4000$            & 2.23 & 450.8  & 0.03\% & 0.04\% & 3213   & 2/18  \\
$5000$            & 2.40 & 469.6  & 0.03\% & 0.15\% & 3542   & 0/18  \\
\bottomrule
\end{tabular}
\begin{tablenotes}
\small
\item \emph{Note.} Top block: medium and large instance sets; each row averages over all $(I_2,m,k)$ configurations as in Table~\ref{tab:AVG_results_ii}. Bottom block: extra-large sets; each row averages over three edge densities (sparse, medium, dense). The column "optimal" reports on how
many instances Gurobi proves optimality within the time limit; otherwise $Z_\mathrm{G}$ is Gurobi's best incumbent, so a negative Gap means that GSS finds a better solution. 
\end{tablenotes}
\end{threeparttable}
\end{table}
As shown in \autoref{tab:superior performance}, GSS matches the solution quality of Gurobi at a fraction of the running time on the medium and large instance sets. For the instance set $S\_800$, GSS attains an average Gap of $0.45\%$ in $0.13$ seconds, against about $75$ seconds for Gurobi, and the efficiency advantage widens with size: Gurobi
proves optimality on all these instances but needs up to almost an hour on the hardest ones, whereas GSS stays below $0.2$ seconds throughout. On instance 900\_0.1\_60\_5\_ii, for example, GSS finds a solution within $0.14\%$  (a DGap of $0.75\%$)  in $0.12$ seconds, while Gurobi takes $3{,}582$ seconds to find the optimal solution. Overall, the Gap of GSS ranges from $0.45\%$ to $1.31\%$ on these sets.

While Gurobi still proves optimality on the medium and large instance sets, it no longer scales to the extra-large networks. \autoref{tab:superior performance} shows that from $2{,}000$ nodes onward, Gurobi almost never proves optimality within the one-hour limit. On some extra-large instances, GSS not only matches but even improves on Gurobi's incumbent: the average Gap is essentially zero at $4{,}000$--$5{,}000$ nodes, and GSS finds a strictly better solution on $7$ of the $90$ instances. \autoref{tab:extralarge_sub} lists these instances; all of them are cases where Gurobi reached the time limit, and the improvement
even reaches $2.1\%$ in terms of DGap (e.g., $8671.4$ versus $8640$ on the instance $4972\_191855\_100\_15\_$ii). Part of this strength comes from the initialization: on the largest
instances, the initial selection of GSS is often already the best
solution found; on 1168\_18126\_100\_3\_ii and
1171\_4738\_100\_5\_ii, for example, the search does not improve the initial solution at all, and it coincides with the proven optimum. 
Moreover, GSS solves every instance,
including networks with more than $5{,}000$ nodes and $190{,}000$
edges, in at most $4.7$ seconds, whereas Gurobi's average running time is at least $50$ minutes on every extra-large set with $|V|\ge 2000$ and, on most instances, reaches the one-hour limit. \autoref{fig:gss_vs_grb} summarizes
the comparison: GSS tracks Gurobi's solution quality at a small
fraction of its running time. In summary, GSS reaches comparable, and even better, solutions in seconds on instances where Gurobi requires the full hour; detailed results can be found in
\ref{appendix:extralarge_ii} (\autoref{tab:extralarge_ii_full}).

\begin{table}[h!]
\centering
{\begin{threeparttable}
\caption{Instances on which GSS strictly outperforms Gurobi
(\textsc{PSIP-DG}$_{\mathrm{B},\infty}$)}\label{tab:extralarge_sub}
\begin{tabular}{@{}lccccc@{}}
\toprule
{} & \multicolumn{4}{c}{GSS} & {GRB} \\
\cmidrule(lr){2-5} \cmidrule(lr){6-6}
Instance & Gap & DGap & $T(s)$ & $N$ & $T(s)$ \\
\midrule
4972\_191855\_100\_15\_ii & $-0.36\%$ & $-2.09\%$ & 3.49 & 503.5 & TL \\
4299\_52560\_100\_10\_ii  & $-0.21\%$ & $-1.44\%$ & 1.21 & 504.0 & TL \\
3977\_130543\_100\_15\_ii & $-0.16\%$ & $-0.90\%$ & 2.11 & 504.4 & TL \\
5260\_76623\_100\_50\_ii  & $-0.11\%$ & $-0.29\%$ & 1.78 & 505.2 & TL \\
2976\_25269\_100\_10\_ii  & $-0.06\%$ & $-0.26\%$ & 0.52 & 500.3 & TL \\
2832\_69332\_100\_15\_ii  & $-0.03\%$ & $-0.17\%$ & 2.66 & 506.7 & TL \\
5260\_76623\_100\_10\_ii  & $-0.03\%$ & $-0.20\%$ & 1.91 & 504.6 & TL \\
\bottomrule
\end{tabular}
\begin{tablenotes}
\small
\item \emph{Note.} A negative Gap/DGap means GSS finds a strictly better
solution than Gurobi's incumbent. Gurobi reaches the one-hour time limit
without proving optimality solution on each listed instance, denoted by "TL".
\end{tablenotes}
\end{threeparttable}}
\end{table}

Overall, SRS obtains higher solution quality than GSS for
\textsc{PSIP-DG}$_{\mathrm{B},1}$, finding an optimal solution on
$85.4\%$ of the instances, at the cost of longer running times on
larger instances. For \textsc{PSIP-DG}$_{\mathrm{B},\infty}$, the
roles reverse: GSS is the more effective method, and on the extra-large instances it reaches solutions  within $0.12\%$ of Gurobi's incumbent on average, finding strictly better ones in some cases, in a few seconds even on networks of up to
$5{,}000$ nodes.
\begin{figure}[H]
\centering
\includegraphics[width=\textwidth]{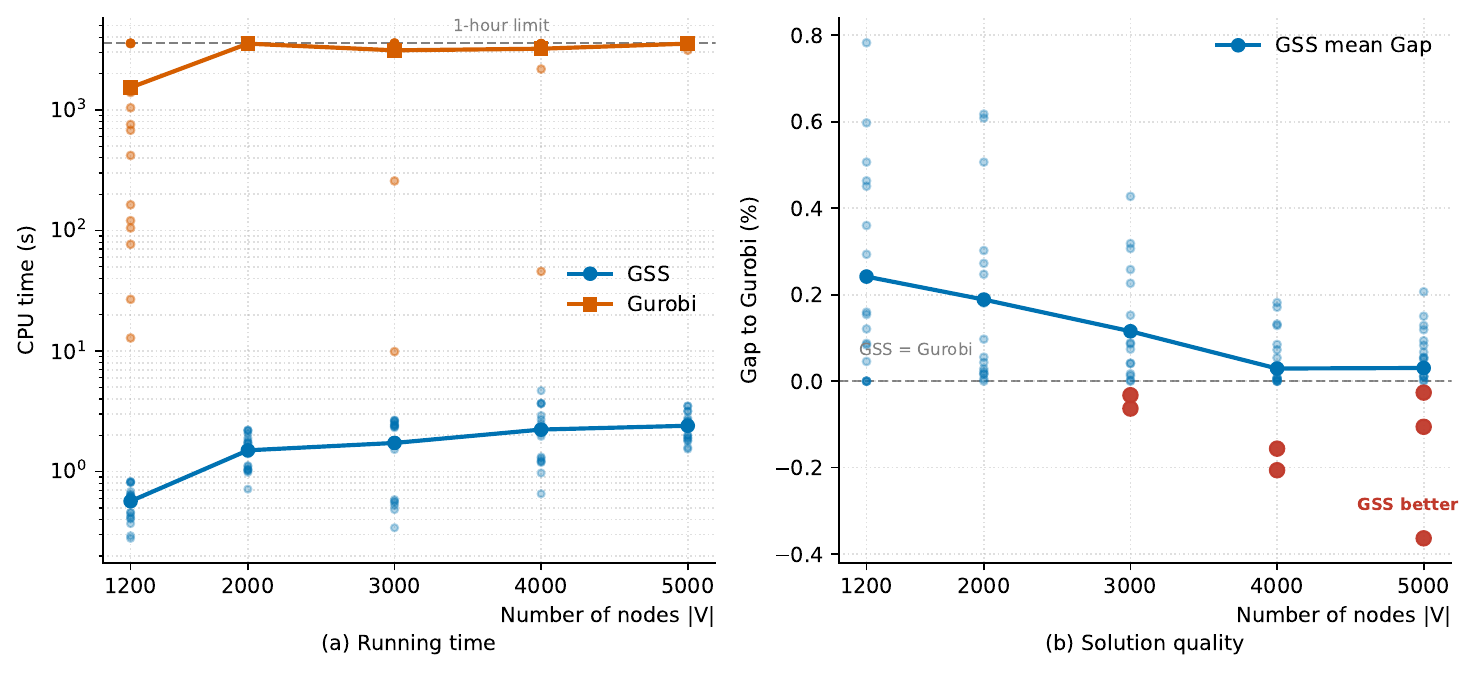}
\caption{Running time (a) and solution quality (b) of GSS
versus Gurobi on the extra-large \textsc{PSIP-DG}$_{\mathrm{B},\infty}$
instances. Faint markers are individual instances, solid lines are
means; in (b), a negative Gap (dark red) means that GSS finds a better
solution than Gurobi's incumbent.}
\label{fig:gss_vs_grb}
\end{figure}

\section{Conclusions}
\label{sec:conclusions}

We study the structural impact of correlated disruptions on  precedence-based project scheduling by introducing delay groups into a network interdiction model. Our model builds on earlier interdiction models \citep{brown2009interdicting} and can serve as a basis for future extensions to more general project settings. The interdictor selects, within a limited budget, a set of delay
groups whose activities are delayed according to given uncertainty
sets, so as to maximize the project's completion time. We investigate the computational complexity of \textsc{PSIP-DG} and its variants: the continuous variant of \textsc{PSIP} with
non-uniform interdiction costs is $N\!P$-hard, \textsc{PSIP-DG} with a general polyhedral uncertainty set is $N\!P$-hard, and the hardness holds for the budgeted uncertainty set, \textsc{PSIP-DG}$_{\mathrm{B}}$, even when at most one activity within each
group can be delayed and when all activities of an attacked group
can be delayed. We further derived an inapproximability bound for the general problem and a greedy $k$-approximation for
\textsc{PSIP-DG}$_{\mathrm{B}}$. On the positive side, we also identified network structures that allow \textsc{PSIP-DG} to be solved efficiently, combining dualization techniques and dynamic
programming.

Based on the fact that the maximum delay of each budgeted
uncertainty group can be computed in polynomial time, we propose two structure-based heuristics for \textsc{PSIP-DG}$_{\mathrm{B}}$: the Subpath Reoptimization Search (SRS), which iteratively improves a candidate path, and the Group Selection Search (GSS), which iteratively improves the set of attacked groups. Our computational study shows that both heuristics attain near-optimal solutions efficiently. SRS achieves the higher solution quality for \textsc{PSIP-DG}$_{\mathrm{B},1}$, whereas GSS performs better for \textsc{PSIP-DG}$_{\mathrm{B},\infty}$, matching the quality of Gurobi's incumbent within a fraction of the running time and  finding even better solutions for some extra-large instances. Although we consider only the extreme cases in which the number of affected activities per group is either one or unrestricted, our experiments indicate that this number strongly influences the performance of the exact solver: as it grows, Gurobi's running time increases at a far higher rate than that of our heuristics, so the good performance of Gurobi on \textsc{PSIP-DG}$_{\mathrm{B},1}$ does not extend to more general settings.

Several open problems arise naturally from this work.  Our hardness results leave a gap between known approximation algorithms and inapproximability bounds. Closing this gap, either by proving tighter lower bounds for specific project network classes (for example series-parallel or sparse PERT networks) or by designing algorithms with improved worst-case guarantees, would be an interesting challenge. Another promising direction is to identify additional polynomially solvable or easily approximable subclasses of \textsc{PSIP-DG}, for instance by bounding delay group sizes or by exploiting special precedence structures. It would also be interesting to extend the framework to alternative objective functions and more complex scheduling environments, for instance weighted completion time, measures of tardiness, or multi-project and resource-constrained settings. The resource-constrained extension is particularly interesting but
raises a non-trivial modeling question: under resource constraints,
the makespan is no longer determined by a longest path, so the inner problem of the interdiction model changes fundamentally. Finally, applying our framework to real project data, and possibly combining it with data-driven or distributionally robust uncertainty sets derived from empirical delay information, would help validate the practical value of delay groups and may suggest further modeling improvements.

\section*{Acknowledgments}
This research was supported by project G072520N of the Research Programme "Optimization and analytics for stochastic and robust project scheduling" of the Fund for Scientific Research—Flanders (Belgium) (F.W.O.-Vlaanderen).
\bibliographystyle{elsarticle-harv} 
\bibliography{reference}
\newpage
\appendix
\section{Algorithm details}\label{appendix: algorithms}

\makeatletter
\renewcommand{\thealgorithm}{A.\arabic{algorithm}}
\makeatother
This appendix provides detailed pseudocode for the dynamic programming subroutines and neighborhood generation procedures used in the heuristic algorithms presented in Section~\ref{sec:heuristics}. In particular, Algorithm~\ref{alg:find_max_delay_path} describes the dynamic programming routine for the singly constrained longest path problem, Algorithms~\ref{alg:dp_group_B1} and~\ref{alg:dp_group_Binf} specify the computation of maximum delay on a fixed path under group-wise budget constraints for \textsc{PSIP-DG}$_{\mathrm{B},1}$ and \textsc{PSIP-DG}$_{\mathrm{B},\infty}$, respectively, and Algorithm~\ref{alg:get_path_neighborhood} details the construction of path neighborhoods exploited by the SRS.

\begin{breakablealgorithm}
\caption{Singly constrained longest path 
$\texttt{find\_max\_delay\_path}(G, s, t, k)$}
\label{alg:find_max_delay_path}
\begin{algorithmic}[1]
\State \textbf{Input:} DAG $G = (V,E)$ with node durations $d_i$ and delay increments $\Delta_i$, 
start node $s$, end node $t$, budget $k$
\State \textbf{Output:} longest path $p^*$ from $s$ to $t$, set of delayed nodes $V^*$, maximum length $l^*$
\State topologically order the nodes in $G$ so that $s$ comes before $t$
\State initialize $f(i,c) \gets -\infty$ for all $i \in V$, $c = 0,\dots,k$
\State initialize $\text{pred}(i,c) \gets$ \texttt{null} and $\text{delayed}(i,c) \gets$ \texttt{false}
\For{$c = 0$ to $k$}
  \State $f(s,c) \gets 0$
\EndFor
\For{each node $j$ in topological order after $s$ up to $t$}
  \For{$c = 0$ to $k$}
    \State $f(j,c) \gets -\infty$, $\text{pred}(j,c) \gets$ \texttt{null}, $\text{delayed}(j,c) \gets$ \texttt{false}
    \For{each predecessor $i$ of $j$}
      \If{$f(i,c) + d_i > f(j,c)$}
        \State $f(j,c) \gets f(i,c) + d_i$
        \State $\text{pred}(j,c) \gets i$, $\text{delayed}(j,c) \gets$ \texttt{false}
      \EndIf
      \If{$c > 0$ \textbf{and} $f(i,c-1) + d_i + \Delta_i > f(j,c)$}
        \State $f(j,c) \gets f(i,c-1) + d_i + \Delta_i$
        \State $\text{pred}(j,c) \gets i$, $\text{delayed}(j,c) \gets$ \texttt{true}
      \EndIf
    \EndFor
  \EndFor
\EndFor
\State $l^* \gets f(t,k)$
\State reconstruct path $p^*$ and delayed set $V^*$ by backtracking from $(t,k)$ using $\text{pred}$ and $\text{delayed}$
\State \textbf{return} $p^*, V^*, l^*$
\end{algorithmic}
\end{breakablealgorithm}

\begin{breakablealgorithm}
\caption{DP for maximum delay on a path $\texttt{DP\_on\_path}(P, k)$ (PSIP-DG$_{\mathrm{B},1}$)}
\label{alg:dp_group_B1}
\begin{algorithmic}[1]
\State \textbf{Input:} path $P$, delay groups $\{V_r\}_{r=0}^{w-1}$, node delays $\Delta_i$, budget $k$

\State \textbf{Output:} maximum delay $\text{max\_delay}$, set of interdicted nodes $S^*$
\State initialize $f(j,c) \gets 0$ and $S(j,c) \gets \emptyset$ for all $j=-1,\dots,w-1$, $c=0,\dots,k$
\For{$j = 0$ to $w-1$}
  \State $C_j \gets V_j \cap P$
  \For{$c = 0$ to $k$}
    \State $f(j,c) \gets f(j-1,c)$, $S(j,c) \gets S(j-1,c)$
    \If{$c > 0$ \textbf{and} $C_j \neq \emptyset$}
      \State choose $i^* \in C_j$ with maximum $\Delta_i$
      \State $\text{val} \gets f(j-1,c-1) + \Delta_{i^*}$
      \If{$\text{val} > f(j,c)$}
        \State $f(j,c) \gets \text{val}$
        \State $S(j,c) \gets S(j-1,c-1) \cup \{ i^* \}$
      \EndIf
    \EndIf
  \EndFor
\EndFor
\State $\text{max\_delay} \gets f(w-1,k)$, $S^* \gets S(w-1,k)$
\State \textbf{return} $\text{max\_delay}, S^*$
\end{algorithmic}
\end{breakablealgorithm}
\begin{breakablealgorithm}
\caption{DP for maximum delay on a path $\texttt{DP\_on\_path}(P, k)$ (PSIP-DG$_{\mathrm{B},\infty}$)}
\label{alg:dp_group_Binf}
\begin{algorithmic}[1]
\State \textbf{Input:} path $P$, delay groups $\{V_r\}_{r=0}^{w-1}$, node delays $\Delta_i$, budget $k$
\State \textbf{Output:} maximum delay $\text{max\_delay}$, mapping $r \mapsto \mathcal{I}_r$ of attacked nodes

\vspace{1mm}
\State initialize $f(j,c) \gets 0$ and $S(j,c) \gets$ empty map for all $j=-1,\dots,w-1$, $c=0,\dots,k$
\For{$j = 0$ to $w-1$}
  \State $C_j \gets V_j \cap P$
  \State $\text{group\_delay} \gets \sum_{i \in C_j} \Delta_i$
  \For{$c = 0$ to $k$}
    \State $f(j,c) \gets f(j-1,c)$, $S(j,c) \gets S(j-1,c)$
    \If{$c > 0$ \textbf{and} $C_j \neq \emptyset$}
      \State $\text{val} \gets f(j-1,c-1) + \text{group\_delay}$
      \If{$\text{val} > f(j,c)$}
        \State $f(j,c) \gets \text{val}$
        \State $S(j,c) \gets S(j-1,c-1)$
        \State $S(j,c)[j] \gets C_j$ 
      \EndIf
    \EndIf
  \EndFor
\EndFor
\State $\text{max\_delay} \gets f(w-1,k)$, mapping $r \mapsto \mathcal{I}_r \gets S(w-1,k)$
\State \textbf{return} $\text{max\_delay}, \{ \mathcal{I}_r \}$
\end{algorithmic}
\end{breakablealgorithm}

\begin{breakablealgorithm}
\caption{Neighborhood of a path $\texttt{get\_path\_neighborhood}(p, k, \texttt{num}, \bar\ell)$}
\label{alg:get_path_neighborhood}
\begin{algorithmic}[1]
\State \textbf{Input:} path $p = (v_0,\dots,v_L)$, budget $k$, segment length $\texttt{num}$,
per-group delay cap $\bar\ell$
\Statex \hspace{1.2em} \Comment{$\bar\ell = 1$ for PSIP-DG$_{\mathrm{B},1}$; \; $\bar\ell = \max_r |V_r|$ for PSIP-DG$_{\mathrm{B},\infty}$}
\State \textbf{Output:} set of neighbor paths $N$, base length $l_p$, base policy $S_p$

\vspace{1mm}
\State $N \gets \emptyset$
\State compute $\text{max\_delay}_p, S_p \gets \texttt{DP\_on\_path}(p,k)$
\State $l_p \gets \sum_{i \in p} d_i + \text{max\_delay}_p$
\State $n_e \gets$ number of edges in $p$
\State $k_{\text{sub}} \gets \min\{k, 2 \cdot \texttt{num}\}$
\Comment{$k'$ counts \emph{groups}; the subpath receives $k'\bar\ell$ delay units}

\vspace{1mm}
\For{$j = 0$ to $n_e - \texttt{num}$}
  \State $s \gets v_j$, $t \gets v_{j + \texttt{num}}$
  \For{$k' = 0$ to $k_{\text{sub}}$}
    \State $(p_{\text{sub}}, V_{\text{sub}}, L_{\text{sub}}) \gets \texttt{find\_max\_delay\_path}(G, s, t, k'\bar\ell)$
    \If{$L_{\text{sub}} < 0$}
      \State \textbf{continue}
    \EndIf
    \State $p_{\text{rem}} \gets$ remaining nodes of $p$ outside segment $[j, j+\texttt{num}]$
    \State compute base length of $p_{\text{rem}}$: $L_{\text{rem}} \gets \sum_{i \in p_{\text{rem}}} d_i$
    \State let $R \gets$ nodes in $p_{\text{rem}}$ sorted in non-increasing order of $\Delta_i$
    \State choose top $(k - k')\bar\ell$ nodes in $R$ and add their delays to $L_{\text{rem}}$   \Comment{optimistic bound}
    \State construct modified path $p' \gets$ prefix of $p$ up to $v_j$ + $p_{\text{sub}}$ + suffix of $p$ after $v_{j+\texttt{num}}$
    \State $L_{p'} \gets L_{\text{sub}} + L_{\text{rem}}$
    \If{$L_{p'} \geq l_p$ and $p' \notin N$}
      \State $N \gets N \cup \{p'\}$
    \EndIf
  \EndFor
\EndFor

\vspace{1mm}
\State \textbf{return} $N, l_p, S_p$
\end{algorithmic}
\end{breakablealgorithm}

\section{Computational study}
\subsection{MILP formulation used with Gurobi}\label{appendix: MIP}
The formulation \ref{dplr} for solving \textsc{PSIP-DG}$_{\mathrm{B}}$ contains products of binary variables and thus leads to a non-linear objective function. To linearize this objective, we follow the arc-splitting technique of \citet{brown2005complexity}, which is also used in their MAXMAX1 formulation. The basic idea is to replace each physical arc by two parallel arcs with fixed lengths, and let the binary decision variables determine which of these arcs can carry flow.

For each arc $(i,j)\in E$ we introduce two flow variables:
$y_{ij}^{(1)}$ and $y_{ij}$.  The variable $y_{ij}^{(1)}$ represents
the default version of arc $(i,j)$ with length $d_i$, while
$y_{ij}$ represents the delayed version of the same arc with
length $d_i + \Delta_i$.  The objective function can then be written as
\begin{equation*}
\label{eq:lin-obj}
\max \quad
\sum_{(i,j)\in E} d_i\, y_{ij}^{(1)}
\;+\;
\sum_{(i,j)\in E} (d_i + \Delta_i)\, y_{ij}.
\end{equation*}
We solve the following MILP formulation for
\textsc{PSIP-DG}$_{\mathrm{B},1}$ and \textsc{PSIP-DG}$_{\mathrm{B},\infty}$ with Gurobi, where the parameter $\ell$ either equals $1$ or is chosen as a sufficiently large number such that $\ell \ge |V_r|$ for all $r \in [m]$.
\begin{alignat*}{2}
	&\max\limits_{y,y^{(1)}}&&\sum_{(i,j)\in E} d_i\, y_{ij}^{(1)}
  + \sum_{(i,j)\in E} (d_i + \Delta_i)\, y_{ij} \\
	&\mbox{s.t.}& \quad 
	&\sum_{i:(i,j)\in E} \begin{aligned}[t]\bigl(y_{ij}^{(1)} + y_{ij}\bigr)
  - \sum_{i:(j,i)\in E} \bigl(y_{ji}^{(1)} + y_{ji}\bigr)&=\begin{cases} 
		-1& \text{if } j=0 \\
		1 & \text{if } j=n+1\\
		0 & \text{if } j\in V\setminus\{0,n+1\}
	\end{cases}&{}&   \\
\sum_{\substack{(i,j)\in E\\ i \in V_r}} y_{ij}
\;&\le\;
\ell \, z_r,   && \forall~ r\in [m]\\
\sum_{r\in [m]}z_r&\leq k &{}&   \\
y_{ij}^{(1)},\, y_{ij}, z_r&\in\{0,1\}  &{}&  \forall~ (i,j)\in E,~ r\in [m]
\end{aligned}
\end{alignat*}

\subsection{Statistics for the test instances}\label{appendix:instances}
\begin{table}[H]
	\centering
	\scalebox{0.8}{\begin{threeparttable}
		\caption{Test problem statistics}
		\label{tab:instances}
		\begin{tabular}{ccccc}
			\toprule
			Category & $|V|$ & $I_2$ & $m$ & $k$ \\
			\midrule
			\multirow{3}{*}{Small-size}
			 & $50$  & $0.1,\,0.5,\,0.8$                                  & $5$--$20$  & $2$--$11$  \\
			 & $100$ & $0.1,\,0.5,\,0.8$                                  & $5$--$50$  & $2$--$25$  \\
			 & $200$ & $0.01,\,0.02,\,0.05,\,0.1,\,0.5,\,0.8$             & $5$--$100$ & $2$--$50$  \\
			\midrule
			\multirow{2}{*}{Medium-size}
			 & $400$ & $0.01,\,0.02,\,0.05,\,0.1,\,0.5,\,0.8$             & $5$--$150$ & $2$--$70$  \\
			 & $600$ & $0.01,\,0.02,\,0.05,\,0.1,\,0.5,\,0.8$             & $5$--$300$ & $2$--$120$ \\
			\midrule
			\multirow{2}{*}{Large-size}
			 & $800$ & $0.05,\,0.1,\,0.5,\,0.8$                           & $5$--$200$ & $2$--$75$  \\
			 & $900$ & $0.01,\,0.02,\,0.05,\,0.1,\,0.2,\,0.4,\,0.5,\,0.8$ & $5$--$300$ & $2$--$100$ \\
			\midrule
			\multirow{5}{*}{Extra-large}
			 & $1200$ & \textemdash & $100$ & $3$--$50$ \\
			 & $2000$ & \textemdash & $100$ & $3$--$50$ \\
			 & $3000$ & \textemdash & $100$ & $3$--$50$ \\
			 & $4000$ & \textemdash & $100$ & $3$--$50$ \\
			 & $5000$ & \textemdash & $100$ & $3$--$50$ \\
			\bottomrule
		\end{tabular}
		\begin{tablenotes}
			\small
			\item \emph{Note.}  The extra-large sets are generated with a layered-DAG generator for which $I_2$ does not apply (\textemdash); instead, three edge
			densities (sparse, medium, dense) are produced per target size, with $m=100$ and $k\in\{3,5,10,15,30,50\}$. For these sets, $|V|$ denotes
			the target size; the actual number of activities varies (e.g., the "$2000$" size ranges from about $2{,}100$ to $2{,}400$ nodes).
		\end{tablenotes}
	\end{threeparttable}}
\end{table}

\subsection{\texorpdfstring{Results for extra-large instances in the case of \textsc{PSIP-DG}$_{\mathrm{B},1}$}{Results for extra-large instances in the case of PSIP-DG-B1}} \label{appendix:extralarge_i}

For the extra-large \textsc{PSIP-DG}$_{\mathrm{B},1}$ instances
(\autoref{tab:AVG_results_large_i}), the initial solutions of both
heuristics are already close to the optimum, and the search phase
changes them very little: the average improvement over the initial
solution is at most $1.00\%$ for GSS and at most $0.10\%$ for SRS. The relaxation underlying the SRS initialization is particularly tight for this variant (Section~4.2.1), so SRS terminates after exploring only a handful of neighbors ($N$ as low as $1$), with a Gap that rounds to $0.00\%$ on every set; GSS is likewise near-optimal (Gap $\le 0.09\%$). Gurobi, however, also solves these instances quickly (on average $2.33$ seconds) and remains the exact reference, so here the heuristics offer no runtime advantage over the solver. The picture changes for \textsc{PSIP-DG}$_{\mathrm{B},\infty}$, where Gurobi fails to prove
optimality, as we discuss next.

\begin{table}[h!]
\centering
\scalebox{1}{\begin{threeparttable}
\caption{Comparison of SRS and GSS on extra-large instances of \textsc{PSIP-DG}$_{\mathrm{B},1}$}
\label{tab:AVG_results_large_i}
\begin{tabular}{@{}lcccccccccc@{}}
\toprule
 & \multicolumn{5}{c}{SRS} & \multicolumn{5}{c}{GSS} \\
\cmidrule(lr){2-6} \cmidrule(lr){7-11}
Instance set & Gap & DGap & $T(s)$ & $N$ & IMP & Gap & DGap & $T(s)$ & $N$ & IMP \\
\midrule
1168\_18126  & 0.00\% & 0.00\% & 2.62  & 19.7 & 0.07\% & 0.03\% & 0.15\% & 0.51 & 507.0 & 0.29\% \\
1171\_4738   & 0.00\% & 0.00\% & 1.70  & 17.7 & 0.00\% & 0.09\% & 0.43\% & 0.24 & 507.0 & 1.00\% \\
1183\_11994  & 0.00\% & 0.00\% & 0.23  & 1.8  & 0.00\% & 0.00\% & 0.00\% & 0.40 & 507.0 & 0.00\% \\
2132\_15791  & 0.00\% & 0.00\% & 5.08  & 17.5 & 0.00\% & 0.05\% & 0.39\% & 0.82 & 508.2 & 0.82\% \\
2322\_57931  & 0.00\% & 0.00\% & 7.41  & 17.5 & 0.00\% & 0.01\% & 0.08\% & 1.43 & 508.4 & 0.23\% \\
2419\_35726  & 0.00\% & 0.03\% & 6.68  & 17.5 & 0.03\% & 0.02\% & 0.18\% & 1.15 & 508.1 & 0.17\% \\
2832\_69332  & 0.00\% & 0.00\% & 0.80  & 1.0  & 0.00\% & 0.00\% & 0.00\% & 2.12 & 508.8 & 0.00\% \\
2976\_25269  & 0.00\% & 0.00\% & 8.86  & 17.5 & 0.00\% & 0.01\% & 0.10\% & 1.22 & 508.3 & 0.21\% \\
3031\_53716  & 0.00\% & 0.03\% & 12.03 & 18.3 & 0.10\% & 0.00\% & 0.00\% & 1.93 & 508.6 & 0.00\% \\
3977\_130543 & 0.00\% & 0.00\% & 2.68  & 2.0  & 0.00\% & 0.00\% & 0.00\% & 4.43 & 509.3 & 0.00\% \\
4169\_91592  & 0.00\% & 0.00\% & 2.85  & 2.0  & 0.00\% & 0.00\% & 0.00\% & 3.88 & 509.1 & 0.00\% \\
4299\_52560  & 0.00\% & 0.00\% & 1.72  & 1.0  & 0.00\% & 0.01\% & 0.18\% & 2.91 & 509.4 & 0.73\% \\
4972\_191855 & 0.00\% & 0.00\% & 2.97  & 1.0  & 0.00\% & 0.00\% & 0.00\% & 7.72 & 509.4 & 0.00\% \\
5120\_140115 & 0.00\% & 0.00\% & 2.89  & 1.0  & 0.00\% & 0.00\% & 0.00\% & 7.02 & 509.4 & 0.00\% \\
5260\_76623  & 0.00\% & 0.00\% & 2.53  & 1.0  & 0.00\% & 0.00\% & 0.00\% & 4.78 & 509.2 & 0.00\% \\
\bottomrule
\end{tabular}
\begin{tablenotes}
\small
\item \emph{Note.} Instance sets are named $|V|\_|E|$. Each row averages over $k\in\{3,5,10,15,30,50\}$, with $m=100$ and one instance per size--density combination. All instances are solved to proven optimality, so gaps are against the true optimum. $T(s)$: average CPU time in seconds; $N$: average number of explored
neighbors; $\mathrm{Gap}=(Z_\mathrm{G}-\mathrm{ALG})/Z_\mathrm{G}$;
$\mathrm{DGap}=(Z_\mathrm{G}-\mathrm{ALG})/(Z_\mathrm{G}-L^*)$;
IMP: average improvement of the search over the initial solution,
$\mathrm{IMP}=(\mathrm{ALG}-\mathrm{ALG}^0)/(\mathrm{ALG}^0-L^*)$.
\end{tablenotes}
\end{threeparttable}}
\end{table}

\subsection{\texorpdfstring{Results for extra-large instances in the case of \textsc{PSIP-DG}$_{\mathrm{B},\infty}$}{Results for
extra-large instances in the case of PSIP-DG-B-inf}}\label{appendix:extralarge_ii}
\begin{table}[htbp]
\centering
\scalebox{1}{\begin{threeparttable}
\caption{Detailed comparison of GSS and Gurobi on the extra-large instances of \textsc{PSIP-DG}$_{\mathrm{B},\infty}$}
\label{tab:extralarge_ii_full}
\begin{tabular}{@{}lcccccc@{}}
\toprule
 & \multicolumn{4}{c}{GSS} & \multicolumn{2}{c}{GRB} \\
\cmidrule(lr){2-5} \cmidrule(lr){6-7}
Instance & Gap & DGap & $T(s)$ & $N$ & $T(s)$ & opt \\
\midrule
1168\_18126  & 0.15\%  & 0.50\%  & 0.72 & 453.0 & 667  & 6/6 \\
1171\_4738   & 0.36\%  & 1.42\%  & 0.38 & 454.2 & 1998 & 3/6 \\
1183\_11994  & 0.21\%  & 0.81\%  & 0.61 & 482.1 & 1932 & 3/6 \\
2132\_15791  & 0.31\%  & 1.25\%  & 1.00 & 483.3 & 3563 & 0/6 \\
2322\_57931  & 0.12\%  & 0.47\%  & 1.92 & 446.3 & 3558 & 0/6 \\
2419\_35726  & 0.14\%  & 0.57\%  & 1.58 & 456.6 & 3501 & 0/6 \\
2832\_69332  & 0.12\%  & 0.55\%  & 2.39 & 463.9 & 3561 & 0/6 \\
2976\_25269  & 0.04\%  & 0.14\%  & 0.51 & 446.6 & 2383 & 2/6 \\
3031\_53716  & 0.18\%  & 0.75\%  & 2.28 & 468.3 & 3405 & 0/6 \\
3977\_130543 & 0.03\%  & 0.03\%  & 2.96 & 451.8 & 3536 & 0/6 \\
4169\_91592  & 0.04\%  & 0.13\%  & 2.63 & 454.5 & 2719 & 2/6 \\
4299\_52560  & 0.02\%  & $-0.04$\% & 1.11 & 446.1 & 3384 & 0/6 \\
4972\_191855 & $-0.00$\% & $-0.01$\% & 3.02 & 465.3 & 3572 & 0/6 \\
5120\_140115 & 0.07\%  & 0.32\%  & 2.36 & 468.6 & 3564 & 0/6 \\
5260\_76623  & 0.03\%  & 0.15\%  & 1.82 & 475.0 & 3489 & 0/6 \\
\bottomrule
\end{tabular}
\begin{tablenotes}
\small
\item \emph{Note.} Instance sets are named $|V|\_|E|$; each row averages over
$k\in\{3,5,10,15,30,50\}$ ($m=100$). ``opt'': number of the six $k$-values for
which Gurobi proved optimality within the one-hour limit. A negative Gap/DGap
(averaged over $k$) means GSS is better than Gurobi on average.
$\mathrm{Gap}=(Z_\mathrm{G}-\mathrm{ALG})/Z_\mathrm{G}$;
$\mathrm{DGap}=(Z_\mathrm{G}-\mathrm{ALG})/(Z_\mathrm{G}-L^*)$;
$T(s)$: average CPU time.
\end{tablenotes}
\end{threeparttable}}
\end{table}

\end{document}